\documentclass[journal]{IEEEtran}

\usepackage{amsmath,amssymb,amsfonts}
\usepackage{amsthm}
\usepackage{algorithmic}
\usepackage{algorithm}
\usepackage{array}
\usepackage{graphicx}
\usepackage{textcomp}
\usepackage{xcolor}
\usepackage{booktabs}
\usepackage{cite}
\usepackage{url}
\usepackage{hyperref}
\usepackage{mathtools}
\usepackage{bm}
\usepackage{bbm}
\usepackage{multirow}
\usepackage{subcaption}

\newcommand{\E}{\mathbb{E}}
\newcommand{\calV}{\mathcal{V}}

\newcommand{\calT}{\mathcal{T}}
\newcommand{\calK}{\mathcal{K}}
\newcommand{\calM}{\mathcal{M}}
\newcommand{\calD}{\mathcal{D}}
\newcommand{\AoI}{A}
\newcommand{\indic}{\mathbbm{1}}
\newcommand{\ho}{\mathrm{ho}}
\newcommand{\safe}{\mathrm{safe}}
\newcommand{\sat}{\mathrm{sat}}
\newcommand{\veh}{\mathrm{veh}}
\newcommand{\ts}{\tau_s}
\newcommand{\td}{\tau_{\mathrm{ac}}}

\theoremstyle{plain}
\newtheorem{theorem}{Theorem}
\newtheorem{lemma}{Lemma}
\newtheorem{corollary}{Corollary}
\newtheorem{proposition}{Proposition}
\theoremstyle{definition}
\newtheorem{definition}{Definition}

\theoremstyle{remark}
\newtheorem{remark}{Remark}
\newtheorem{assumption}{Assumption}

\title{Safety-Aware AoI Scheduling for LEO Satellite-Assisted Autonomous Driving}

\author{Kangkang Sun,~\IEEEmembership{Member,~IEEE,}
Junyi He,~\IEEEmembership{Member,~IEEE,} Juntong Liu,
        Xiuzhen Chen,~\IEEEmembership{Member,~IEEE,}
        Jianhua Li,~\IEEEmembership{Senior Member,~IEEE,}
        Minyi Guo,~\IEEEmembership{Fellow,~IEEE}
  \thanks{Kangkang Sun, Juntong Liu, Jianhua Li, Xiuzhen Chen and Minyi Guo are
          with the Shanghai Key Laboratory of Integrated Administration
          Technologies for Information Security, School of Computer Science,
          Shanghai Jiao Tong University, Shanghai 200240, China
          (e-mail: szpsunkk@sjtu.edu.cn; ljttt0824@sjtu.edu.cn; lijh888@sjtu.edu.cn;
          xzchen@sjtu.edu.cn; guo-my@cs.sjtu.edu.cn).}
  \thanks{Junyi He is  School of Science and Engineering, The Chinese University of Hong Kong, Shenzhen, Guangdong, China (e-mail: junyihe@link.cuhk.edu.cn).}
  \thanks{\textit{Corresponding author: Jianhua Li}.}
  \thanks{This work has been submitted to the IEEE for possible publication. Copyright may be transferred without notice, after which this version may no longer be accessible.}
}

\begin{document}

\maketitle

\begin{abstract}
Autonomous platoons traversing infrastructure gaps increasingly depend
on LEO satellite backhaul for safety-critical updates, yet no existing
framework jointly addresses compound Doppler from simultaneous satellite
and vehicle motion, sub-slot handover outages that exceed
collision-alert deadlines, and heterogeneous freshness requirements
across three vehicular priority classes. The core challenge is a
\emph{timescale mismatch}: coarse control slots hide sub-slot outages,
which makes both AoI spike analysis and safety verification ill-posed.
Ping-pong handover oscillations further compound AoI cost in a way that
purely reactive schedulers cannot mitigate. We address these challenges
through a unified framework that couples a two-timescale AoI model with
tiered time-average safety constraints enforced by virtual queues. A
closed-form ping-pong AoI envelope reveals that cumulative penalty
grows \emph{quadratically} in oscillation length, analytically
justifying oscillation suppression as the highest-leverage safety
mechanism. The resulting drift-plus-penalty template is instantiated as
SafeScale-MATD3 with proactive handover timing and multi-task
dual-critic MARL. A key finding is that suppressing brief but repeated
ping-pong oscillations yields larger safety returns than shortening any
single outage, and that tick-level AoI accounting is a necessary
condition for verifiable collision-alert guarantees under LEO
handovers. Simulations show that SafeScale-MATD3 is the only method
satisfying the strict 1\% collision-alert violation budget, reducing
violation rate by 4 to 5.5 times versus baselines, while achieving
35\% lower collision-alert AoI and strict Pareto dominance on the
energy and freshness tradeoff.
\end{abstract}


\section{Introduction}
\label{sec:intro}

\IEEEPARstart{A}{utonomous} driving safety depends on timely
exchange of perception and intent. In remote highways and
infrastructure gaps, satellite backhaul is often the only wide-area
link~\cite{al2022survey, jiang2024game, wang2025integrated, wang2025low}. Low Earth orbit (LEO) constellations such
as Starlink offer global coverage and low propagation delay, yet a
fundamental tension persists: \emph{safety-critical vehicular traffic
demands sub-second information freshness, while LEO handovers impose
repeated link interruptions that can outlast the very deadlines they
must serve.}

The \emph{Age of Information} (AoI)~\cite{yates2021age}, defined
as the elapsed time since the last successfully received update was
generated, captures information staleness in ways that throughput
or delay alone cannot. In platoon-based autonomous driving, this
distinction is safety-critical: collision alerts require near-instant
freshness for emergency braking, platoon control messages demand
tighter timeliness for string stability~\cite{parvini2021aoi}, and
map updates tolerate substantially longer staleness. These three
priority classes carry fundamentally different violation tolerances,
yet no existing LEO AoI framework formulates or enforces them jointly.
Three structural gaps prevent prior work from meeting these
requirements simultaneously.

There are three gaps that prevent prior work from meeting these requirements simultaneously:
\textit{Gap 1: dual-dynamic channel modeling is absent.}
All existing LEO AoI works, including multi-hop relay
analysis~\cite{chiariotti2020aoi,soret2020latency, vikhrova2020age}, DPP-based downlink
scheduling~\cite{dai2025amdt, lei2025joint, jiao2023age}, and SAGIN joint
optimization~\cite{lang2025aoi, zhang2025aoi, huang2025task}, assume static ground users
and model only satellite-side Doppler. When both the satellite
and the vehicle move simultaneously, a compound Doppler shift arises
that invalidates quasi-static channel assumptions and collapses
coherence time to a small fraction of any scheduling slot.
As we derive in Section~\ref{sec:doppler}, this distinguishes
LEO-AV systems structurally from all prior work.
Vehicular AoI studies~\cite{parvini2021aoi,zhang2024drl,azizi2024vlc}
address vehicle dynamics but operate over terrestrial links with
negligible handover cost, leaving the compound Doppler problem
entirely unaddressed.
\textit{Gap 2, timescale inconsistency corrupts AoI spike modeling.}
Prior LEO AoI work uses coarse control slots whose duration is much
longer than a typical handover outage~\cite{dai2025amdt,lang2025aoi}.
Because outages are shorter than one slot, any slot-granularity
floor operation maps the outage to zero, yet a continuous
slot-fraction surrogate remains positive and misleading. Mixing
these two representations inflates spike estimates and leaves
sub-slot safety checks ill-defined~\cite{park2025joint}. Without a
finer tracking tick aligned to actual outage duration, no scheduler
can produce verifiable safety guarantees for millisecond-grade
freshness requirements.
\textit{Gap 3, safety-critical multi-priority AoI constraints
are unformulated.}
No prior LEO AoI work jointly maps heterogeneous vehicular safety
budgets, from tight collision-alert tolerances to permissive
map-update tolerances, into enforceable time-average constraints
with handover-aware online control. Existing LEO AoI
schedulers~\cite{dai2025amdt,lang2025aoi} minimize average AoI
without priority differentiation, while vehicular
schedulers~\cite{parvini2021aoi,zhang2024drl} capture multi-priority
structure but ignore LEO handover dynamics entirely. This orthogonal
blind spot means that even a single handover outage already pushes
collision-alert AoI beyond the safety threshold, a violation that no
purely reactive scheduler can prevent, regardless of its
inter-handover optimality.
The three gaps above motivate the following questions:

\begin{itemize}
\item \textit{Q1: (Gaps~1\,\&\,2, Sec.~\ref{sec:model}): How should AoI be modeled to \emph{simultaneously} capture
sub-slot outage durations, multi-priority safety thresholds, and
coarse control slots, without timescale inconsistency?}

\item \textit{Q2: (Gaps~1\,\&\,2, Sec.~\ref{sec:analysis}): How does cumulative AoI cost scale with the length of ping-pong handover sequences, and why is suppressing oscillations more valuable than shortening individual outages?}

\item \textit{Q3: (Gap~3, Sec.~\ref{sec:problem}): How can time-average AoI violation budgets be enforced \emph{online}
with provable guarantees, without per-slot distributional assumptions?}

\item \textit{Q4: (Gaps~1 to 3, Sec.~\ref{sec:algorithm}): How should this multi-timescale, safety-constrained scheduling problem
be solved efficiently in a multi-platoon, multi-satellite setting?}

\end{itemize}


We address the four questions through a unified framework combining a
two-timescale AoI model, handover spike analysis, drift-plus-penalty
(DPP) virtual queue control, and a new multi-agent reinforcement
learning algorithm, SafeScale-MATD3. The main contributions are as
follows.

\begin{itemize}
    \item \textit{Two-timescale AoI model with compound Doppler (Q1).}
    Existing LEO AoI models use slot-level granularity that maps
    sub-slot handover outages to zero, which makes safety verification
    ill-posed. We resolve this by introducing a fine-grained AoI
    accounting tick nested inside a coarser control slot, ensuring
    that both outage lengths and safety thresholds resolve to non-zero
    integers. A compound Doppler analysis further reveals that
    satellite motion dominates coherence-time variation, enabling
    TLE-based orbital prediction as a low-dimensional channel
    predictor. This two-timescale structure is a \emph{prerequisite}
    for all subsequent safety results.

    \item \textit{Quadratic ping-pong spike envelope (Q2).}
    Quantifying the AoI cost of consecutive handovers is analytically
    challenging because each outage elevates the baseline for the next,
    creating a compounding effect invisible to per-outage analysis. We
    derive a closed-form envelope showing that cumulative AoI penalty
    grows \emph{quadratically} in ping-pong sequence length, a
    counter-intuitive result implying that a burst of $k{=}3$
    consecutive handovers is $6{\times}$ costlier than three isolated
    ones. This insight shifts the design priority from shortening
    individual outages to suppressing oscillation sequences, and
    analytically justifies the discretionary handover budget in our
    formulation.

    \item \textit{Tiered safety enforcement via virtual queues (Q3).}
    Enforcing heterogeneous freshness budgets, 1\% for collision
    alerts, 5\% for platoon control, and 20\% for map updates, online
    is difficult because the AoI distribution under time-varying
    handovers lacks a tractable closed form. We map the three budgets
    to stabilizable virtual queues and prove that strong stability
    implies time-average violation compliance \emph{without
    distributional assumptions}; only bounded per-slot increments and a
    Slater-feasible baseline suffice. The virtual-queue backlog serves
    as a runtime-observable tightness signal that the policy exploits
    for anticipatory control.

    \item \textit{SafeScale-MATD3, a structurally safe multi-agent scheduler (Q4).}
    Instantiating the DPP template as a practical MARL algorithm
    requires reconciling cooperative interference management with
    priority-specific safety enforcement, a multi-objective challenge
    that reward shaping alone cannot resolve. Our key design insight is
    that each virtual queue maps to a dedicated critic module, making
    safety enforcement \emph{structurally embedded} rather than
    reward-engineered. Evaluation shows that SafeScale-MATD3 is the
    \emph{only} method satisfying the strict 1\% collision-alert
    budget, reducing the violation rate by $4{\times}$ to
    $5.5{\times}$ versus all baselines while achieving a 35\% lower
    collision-alert AoI and strict Pareto dominance on the energy and
    freshness tradeoff.
\end{itemize}


\section{Related Work}
\label{sec:related}

This paper sits at the intersection of LEO-satellite AoI analysis and
vehicular AoI scheduling, two research lines that have largely
progressed in isolation. We organize the review along these two lines:
Section~\ref{sec:rw_leo} surveys AoI optimization in LEO satellite
networks, and Section~\ref{sec:rw_v2x} reviews AoI scheduling in
vehicular networks. 

\begin{table}[t]
\centering
\caption{Positioning vs.\ related work (qualitative).}
\label{tab:rw_compare}
\footnotesize
\resizebox{0.5\textwidth}{!}{
\begin{tabular}{lcccccc}
\toprule
\textbf{Ref.} & \textbf{Compound Doppler} & \textbf{HO model} &
\textbf{Safety budget} & \textbf{Multi-pri.} & \textbf{MARL} &
\textbf{Two-scale}\\
\midrule
\cite{chiariotti2020aoi} & N/A & coarse & N/A & N/A & N/A & N/A\\
\cite{dai2025amdt} & sat.\ only & N/A & N/A & N/A & DPP & slot\\
\cite{lang2025aoi} & sat.\ only & yes & N/A & N/A & DQN & slot\\
\cite{park2025joint} & N/A & yes & N/A & N/A & MPC & N/A\\
\cite{parvini2021aoi} & N/A & N/A & N/A & yes & TD3 & slot\\
\textbf{Ours} & yes & yes & yes & yes & TD3 & tick+slot\\
\bottomrule
\end{tabular}}
{\footnotesize
Compound Doppler: ``sat.\ only'' means only satellite-side Doppler is
modeled. Safety budget: explicit AoI safety-violation constraints.
Two-scale: ``tick+slot'' denotes sub-slot AoI accounting plus slot-level
decision control.}
\end{table}

\subsection{AoI in LEO Satellite Networks}
\label{sec:rw_leo}

LEO-AoI studies focus on relay-chain freshness and
single-satellite scheduling but do not model mobile users or
handover-induced spike costs. Chiariotti~et~al.~\cite{chiariotti2020aoi} derived tight
hypoexponential AoI bounds for multi-hop LEO relay chains under
FCFS/OPF/HAF policies, providing a phase-segmented decomposition
framework. However, they assume static topology and do not model
handover.
Dai~et~al.~\cite{dai2025amdt} proposed Lyapunov DPP-based AoI-aware
downlink scheduling with angular beamforming under Shadowed-Rician
fading. Their setting remains single-satellite/multi-user without
handover modeling, heterogeneous safety budgets, or priority-specific
constraints.
Lang~et~al.~\cite{lang2025aoi} combined diffusion-augmented DQN with
handover-frequency minimization in HAP-assisted SAGIN, representing a
nearby AoI baseline with mobility-aware handover control. Their setting
still assumes static ground users and lacks multi-priority safety
budgeting.
Park~et~al.~\cite{park2025joint} measured Starlink handover outages
(outage duration consistent with LEO handover measurements) and proposed joint pruned
MPC for video QoE under handover dynamics.
LEO-AoI studies model orbital/channel dynamics but generally assume
stationary terminals. None jointly captures compound
vehicle and satellite Doppler, sub-slot handover spike analysis, and
priority-differentiated safety constraints
(Table~\ref{tab:rw_compare}, first three columns).
In short, existing LEO-AoI work does not yet supply a modeling stack
that is simultaneously mobile-terminal aware, handover spike faithful,
and safety budget oriented; the present paper targets exactly that
stack.

\subsection{AoI in Vehicular Networks}
\label{sec:rw_v2x}

Vehicular-AoI studies address multi-priority freshness and
MARL-based scheduling but operate over terrestrial links with
negligible satellite handover cost.
Parvini~et~al.~\cite{parvini2021aoi} introduced a global and local
dual-critic TD3 architecture for platoon C-V2X AoI scheduling, showing
that task-decomposed critics reduce gradient conflict across flows.
Their formulation operates over terrestrial sidelinks with negligible
handover cost.
Zhang~et~al.~\cite{zhang2024drl} analyzed multi-priority queues
(HPD$>$DENM$>$CAM) in C-V2X Mode~4 with NOMA, reporting substantial AoI
gains for high-priority messages via hybrid-action DRL. Their priority
classification motivates our message hierarchy but assumes fixed
terrestrial coverage.
Azizi~et~al.~\cite{azizi2024vlc} proposed MA-TD3 for energy-efficient
AoI optimization in VLC-V2X and showed rapid convergence, supporting
TD3-style multi-agent control for mixed objectives.
In mobility and handover management, bandit-based handover
thresholding~\cite{BaTT} addresses the ping-pong versus late-handover
trade-off under extreme mobility. Our proactive TLE-based module
(Section~\ref{sec:proactive}) complements such threshold tuning with
queue-informed timing.
Prediction-based integrated satellite and terrestrial association with
connection-change penalties~\cite{IntSatTerr} shares the spirit of using
orbital prediction to reduce handover cost. We focus on AoI spike
quantification and multi-priority safety enforcement rather than
throughput-optimal association alone.
Separately, game-theoretic satellite resource management
(surveys~\cite{jiang2024game,luong2025incentive}, coalition
games~\cite{kondo2025coalition,li2026coalition}) studies cooperative
spectrum/beam allocation, but does not model AoI evolution with safety
constraints, and it is complementary to this work.
In summary, the vehicular AoI line provides multi-priority MARL
machinery but leaves open how freshness metrics and coordination survive
once LEO handovers and sub-slot AoI evolution enter the picture
(Table~\ref{tab:rw_compare}, handover and two-scale columns).


\section{System Model, Two-Timescale Design, and Problem Formulation}
\label{sec:model}

This section specifies the LEO-assisted platoon scenario and formulates
the joint optimization problem. In plain terms, platoon leaders maintain
satellite links while followers share updates over short-range V2V,
handover outages occur in the hundred-millisecond range, and three
message classes differ sharply in allowable staleness. We organize this
section as follows: Section~\ref{sec:timescales} introduces the
two-timescale structure; Section~\ref{sec:doppler} derives the compound
Doppler model; Section~\ref{sec:handover_model} formalizes handover and
tick-level AoI evolution; Section~\ref{sec:e2e} completes power and
end-to-end AoI accounting; Section~\ref{sec:problem} states the final
constrained optimization.

Three foundational modeling choices jointly determine the analytical
architecture, each necessitated by one of the gaps in
Section~\ref{sec:intro}. \textbf{Choice 1} (Gap~2): the two-timescale
structure ($\ts$, $\td$) resolves the floor-approximation inconsistency.
\textbf{Choice 2} (Gap~1): compound Doppler analysis establishes
satellite dominance and motivates TLE prediction as scheduler state.
\textbf{Choice 3} (Gap~3): discretionary/forced handover splitting
ensures handover-budget feasibility by construction.
These three choices map directly to the three structural features of
the LEO-platoon scenario in Fig.~\ref{fig:system}: compound Doppler
dynamics, mandatory handover outages that exceed collision-alert
deadlines, and order-of-magnitude deadline differences across message
classes. The following subsections formalize each feature.
Fig.~\ref{fig:system} illustrates the complete system scenario.
A fleet of $P$ autonomous vehicle platoons traverses a remote
highway segment without terrestrial cellular infrastructure. Each
platoon is led by a Platoon Leader (PL) maintaining a satellite uplink
to the LEO constellation; up to $N_j{-}1$ followers receive
safety-critical updates through intra-platoon V2V Mode~4 links.

\begin{figure}[t]
    \centering
    \includegraphics[width=0.45\textwidth]{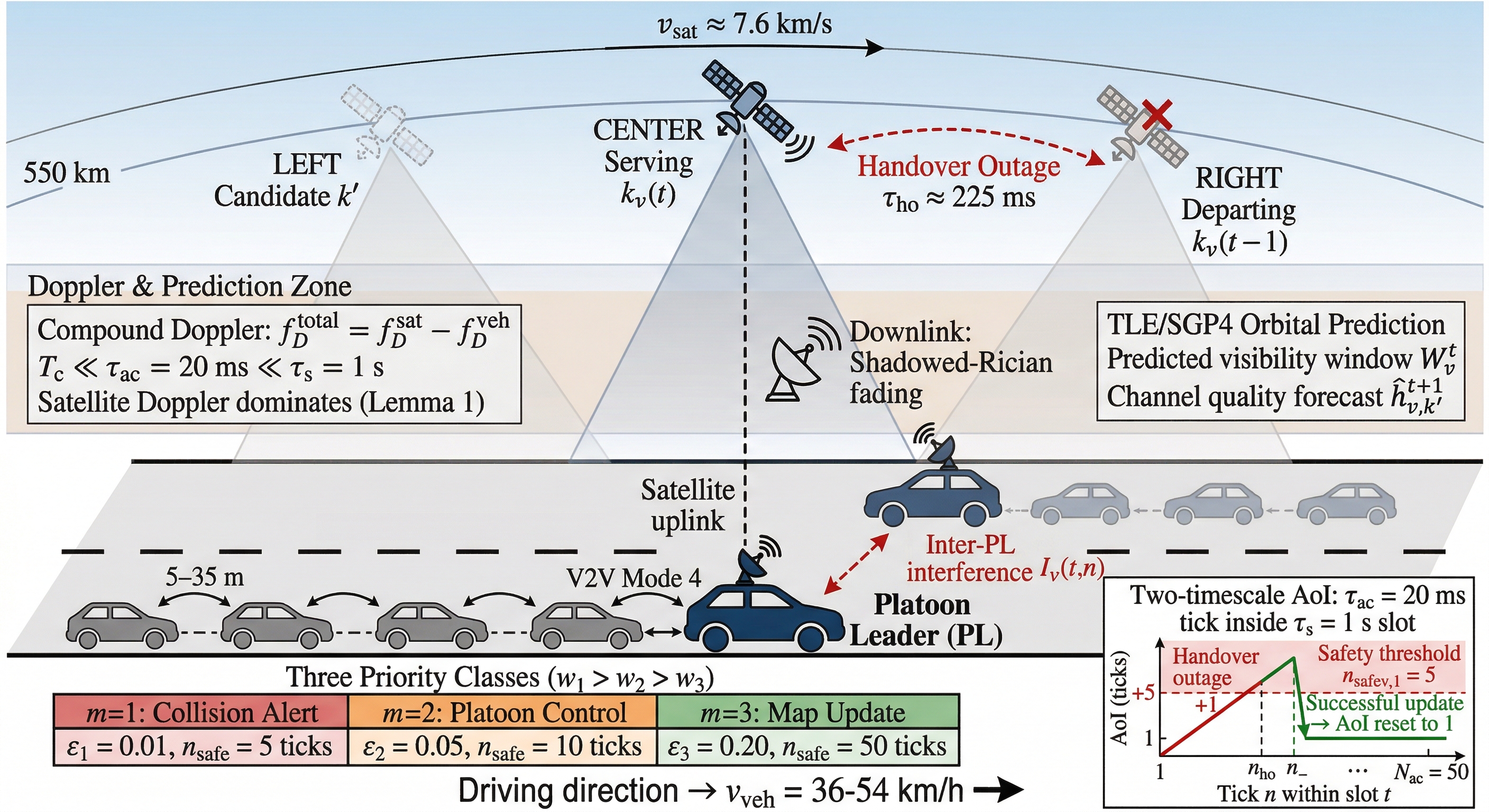}
    \caption{System overview of LEO satellite-assisted autonomous
    vehicle platoon communication under the two-timescale AoI
    framework. 
    }
    \label{fig:system}
\end{figure}

\subsection{Two-Timescale Rationale and Network Architecture}
\label{sec:timescales}

In this subsection, we resolve the time-scale inconsistency.
Three operational timescales must be simultaneously resolved: (i) the
per-class safety freshness thresholds $\Delta_{v,m}^\safe$; (ii) the
handover outage duration $\tau_\ho$~\cite{park2025joint}, which is
shorter than a control slot yet long relative to safety ticks; and
(iii) the control/RL slot $\ts$ matching the DPP
literature~\cite{dai2025amdt,lang2025aoi}. Because $\tau_\ho < \ts$,
the slot-granularity floor maps every outage to zero, yet the
continuous surrogate $\mu_\tau/\ts$ is positive and misleading.

We introduce a fine-grained \emph{AoI
accounting tick} $\td$ alongside the decision slot $\ts$, chosen so
that safety thresholds $n_{v,m}^\safe = \Delta_{v,m}^\safe/\td$ and
outage lengths $n_\ho = \lfloor\tau_\ho/\td\rfloor$ are
\emph{non-zero} integers. The floor expectation relative error
$|\lfloor\tau/\td\rfloor - \tau/\td|/n_\ho$ is then bounded and
explicit. AoI state and safety checks evolve at $\td$-resolution;
satellite association, scheduling, and RL actions occur at
$\ts$-resolution.
We adopt a Walker-Delta model
$\alpha$\cite{choi2026ilcho}: $ N_\mathrm{total}/N_p/F$ at altitude $H_L=550$\,km.
The $m$-th satellite position at slot $t$
is
\begin{equation}
\mathbf{s}_m^t = a_s
\begin{pmatrix}
\cos\psi_m^t \cos\Omega_m - \sin\psi_m^t \sin\Omega_m \cos\alpha \\
\cos\psi_m^t \sin\Omega_m + \sin\psi_m^t \cos\Omega_m \cos\alpha \\
\sin\psi_m^t \sin\alpha
\end{pmatrix},
\label{eq:sat_position}
\end{equation}
where $\psi_m^t=\Phi_m+\omega t+f_0$, $a_s=R_E+H_L$,
$\omega=\sqrt{GM_E/a_s^3}$, and $\Omega_m$ is the RAAN.

\begin{proposition}[Tick-size selection criterion]
\label{prop:tick_selection}
Let $n_1^\safe=\Delta_{\safe,\min}/\td$ be the minimum-class safety
threshold in ticks, $n_\ho=\lfloor\tau_{\ho,\min}/\td\rfloor$ the
minimum handover-outage length in ticks, and
$N_\mathrm{ac}=\ts/\td$ the number of AoI ticks per decision slot.
A sufficient design rule that guarantees $n_1^\safe\ge 2$ and
$n_\ho\ge 1$ is
\begin{equation}
\td \le \frac{\min(\Delta_{\safe,\min},\tau_{\ho,\min})}{2}.
\label{eq:tick_rule}
\end{equation}
The factor of~$2$ is sufficient but not necessary: $n_\ho\ge 1$
requires only $\td\le\tau_{\ho,\min}$, and $n_1^\safe\ge 2$ requires
only $\td\le\Delta_{\safe,\min}/2$. The joint criterion with the
factor-$2$ denominator simultaneously ensures both with a single
inequality and provides a one-tick resolution margin against rounding.
Here $\tau_{\ho,\min}$ denotes the minimum handover outage duration;
because actual outage durations $\tau_\ho^{(t)}$ are random, we
interpret $\tau_{\ho,\min}$ as the support lower bound of their
distribution (i.e., $\tau_\ho^{(t)}\ge\tau_{\ho,\min}$ almost
surely), consistent with the statistical characterization
in~\cite{park2025joint}. The specific value of $\td=20$\,ms is
fixed to satisfy~\eqref{eq:tick_rule} while balancing computational
cost and resolution margin.
\end{proposition}

\begin{proof}
Since $\tau_\ho^{(t)}\ge\tau_{\ho,\min}$ almost surely,
$n_\ho=\lfloor\tau_\ho^{(t)}/\td\rfloor\ge
\lfloor\tau_{\ho,\min}/\td\rfloor\ge 1$ whenever
$\td\le\tau_{\ho,\min}$, which is implied by~\eqref{eq:tick_rule}.
Similarly, $n_1^\safe=\Delta_{\safe,\min}/\td\ge 2$ whenever
$\td\le\Delta_{\safe,\min}/2$, also implied by~\eqref{eq:tick_rule}.
Taking $\td\le\min(\Delta_{\safe,\min},\tau_{\ho,\min})/2$ satisfies
both conditions simultaneously.
\end{proof}

\begin{figure}[t]
\centering
\includegraphics[width=\linewidth]{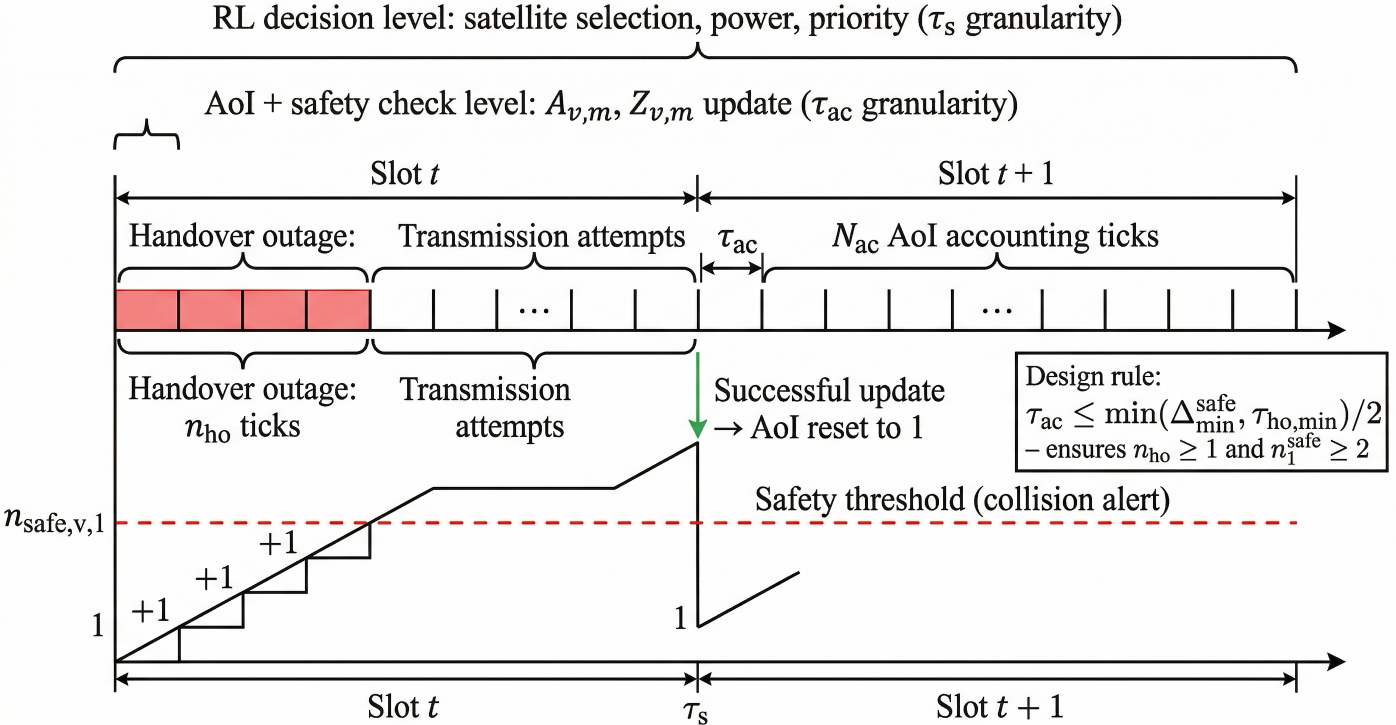}
\caption{Two-timescale AoI framework (Proposition~\ref{prop:tick_selection}).
}
\label{fig:twotimescale}
\end{figure}

Fig.~\ref{fig:twotimescale} illustrates the nested two-timescale
structure that underpins the entire framework. At the outer level, RL
decisions are made once per slot~$\ts$ with full satellite association,
power, and priority control. At the inner level, AoI
state~$A_v^{(t,n)}$ and safety virtual queues~$Z_{v,m}$ evolve at
every tick~$\td$, capturing sub-slot events that a coarse slot model
would miss. In particular, the handover outage phase (red ticks) maps
to a non-zero integer
$n_{\mathrm{ho}}=\lfloor\tau_{\mathrm{ho}}/\td\rfloor\ge 1$, and the
collision-alert threshold resolves to $n_{v,1}^{\mathrm{safe}}\ge 2$,
both guaranteed by the design criterion in Eq.~\eqref{eq:tick_rule}. This
separation ensures that safety violations are detected within the tick
at which they occur, rather than being averaged away at the slot
boundary.
Decision slots $\calT=\{1,2,\ldots,T\}$ have duration $\ts$. Each
slot contains $N_\mathrm{ac}=\ts/\td$ AoI ticks. The system
comprises:
\begin{itemize}
    \item \textbf{LEO constellation:} $\calK=\{1,\ldots,K\}$
    satellites in low Earth orbit. Visible set
    $\calK(t)\subseteq\calK$ updates deterministically per
    TLE/SGP4.

    \item \textbf{Vehicle platoons:} $P$ platoons, each with one
    Platoon Leader (PL) and $N_j-1$ followers. Let
    $\calV=\{v_1,\ldots,v_{N_v}\}$ denote PLs with $N_v=P$. PLs
    maintain the satellite uplink and followers receive V2V CAM over
    Mode~4.

    \item \textbf{Message classes:}
    $m\in\calM=\{1,2,3\}$: collision alerts (HPD-type),
    platoon control (DENM-type), map updates (CAM-type).
\end{itemize}

\subsection{Dual-Dynamic Doppler and Channel Model}
\label{sec:doppler}

With the two-timescale structure in place, we now characterize the
channel dynamics that determine per-tick transmission success.
The key distinction from all prior LEO-AoI work is that
\emph{both} the satellite and the vehicle move simultaneously,
producing a compound Doppler shift that collapses coherence time
well below a single AoI tick.

\begin{lemma}[Dual-dynamic Doppler and coherence time]
\label{lem:doppler}
Let $\theta_{\sat}(t)\in[0^\circ,90^\circ]$ be the satellite
elevation angle and $\psi_{\veh}(t)$ the angle between the vehicle
velocity vector and the satellite--vehicle link. The compound Doppler
shift is
\begin{equation}
f_D^\mathrm{total}(t) = \frac{f_c}{c}\bigl[
v_{\sat}\sin\theta_{\sat}(t)
- v_{\veh}\cos\psi_{\veh}(t)\bigr],
\label{eq:doppler_exact}
\end{equation}
and the channel coherence time is
\begin{equation}
T_c(t) \approx \frac{0.423}{|f_D^\mathrm{total}(t)|}.
\label{eq:coherence}
\end{equation}
Because $v_{\sat}\gg v_{\veh}$, we have
$f_D^\mathrm{total}\approx f_D^\sat$ and $T_c$ is governed
by satellite geometry. The resulting $T_c$ is far shorter than
both~$\ts$ and~$\td$ (verified in simulation), so the channel
undergoes many independent fading realizations within one AoI tick.
\end{lemma}

\begin{corollary}[Doppler dominance and nonstationarity index]
\label{cor:doppler_dominance}
Define the Doppler ratio (valid for $\sin\theta_{\sat}>0$):
\begin{equation}
\rho_D \triangleq \left|\frac{f_D^\veh}{f_D^\sat}\right|
\le \frac{v_{\veh}}{v_{\sat}\sin\theta_{\sat}}.
\label{eq:rho_D}
\end{equation}
Under the elevation mask $\theta_{\sat}\ge\theta_{\min}>0$
(Section~\ref{sec:handover_model}), $\rho_D \ll 1$: satellite
motion dominates Doppler while vehicle speed mainly affects the
effective update rate~$\lambda_{\mathrm{eff}}(v)$.
The nonstationarity index
$\mathrm{NSI}(\tau)\triangleq \tau/T_c$ satisfies
$\mathrm{NSI}(\td)\gg 1$ and $\mathrm{NSI}(\ts)\gg 1$,
quantifying strong nonstationarity at both timescales.
\end{corollary}

Satellite Doppler dominance has two direct design consequences:
\emph{(i)}~TLE/SGP4 orbital predictions serve as a low-dimensional
channel predictor because coherence-time trends track satellite
geometry;
\emph{(ii)}~since $T_c\ll\td$, each tick averages over many
independent fading draws, motivating diffusion-based latent
augmentation (Section~\ref{sec:encoding}) to reduce critic variance.

Let $L_{v,k}(t)$ denote the large-scale path gain from satellite~$k$
to vehicle~$v$ at slot~$t$, $\rho_{v,k}(t,n)$ the small-scale fading
coefficient at tick~$n$, and $\mathbf{a}(\theta,\phi)$ the UPA
steering vector. The composite channel is
\begin{equation}
h_{v,k}(t,n) = \sqrt{L_{v,k}(t)}\,\rho_{v,k}(t,n)\,
\mathbf{a}(\theta_{v,k},\phi_{v,k})\,
e^{j2\pi f_D^\mathrm{total}(t)\cdot n\td},
\label{eq:channel}
\end{equation}
where the exponential captures \emph{residual} Doppler after bulk
pre-compensation~\cite{dai2025amdt}, and
$|\rho_{v,k}(t,n)|$ follows a Shadowed-Rician distribution with
light-shadowing parameters. The large-scale gain uses an NTN-style
decomposition:
\begin{equation}
\begin{aligned}
    P_{\mathrm{loss},v,k}&=F_{\mathrm{loss}}(d_{v,k},f_c)
    +\mathrm{SF}_{v,k}+\frac{L_{\mathrm{zenith}}(f_c)}{\sin\theta_{v,k}}, \\
    L_{v,k}(t) &= P_{\mathrm{loss},v,k}^{-1}.
\end{aligned}
\label{eq:pathloss}
\end{equation}

\begin{assumption}[Shadowed-Rician fading]\label{asm:fading}
$|\rho_{v,k}(t,n)|$ follows a Shadowed-Rician distribution with
light-shadowing parameters, consistent with LEO NTN channel
measurements.
\end{assumption}

\begin{assumption}[Independent tick-level fading]\label{asm:independence}
Fading coefficients $\rho_{v,k}(t,n)$ are independent across
ticks~$n$. This is justified by $T_c\ll\td$
(Lemma~\ref{lem:doppler}): the channel traverses many coherence
intervals per tick, making inter-tick correlation
$\mathcal{O}(T_c/\td)\ll 1$ negligible.
\end{assumption}

\begin{assumption}[Deterministic V2V one-hop delay]\label{asm:v2v}
V2V Mode-4 intra-platoon delay is modeled as a deterministic one-hop
latency dominated by the 20\,ms AoI tick, consistent with C-V2X
sidelink specifications.
\end{assumption}

Let $\gamma_{\mathrm{th}}$ be the minimum SNR for successful decoding.
The per-tick SINR under inter-PL interference is
\begin{equation}
    \begin{aligned}
        &\gamma_{v,k}(n) =
        \frac{P_v|\mathbf{w}_v^H\mathbf{h}_{v,k}(t,n)|^2}
        {\sigma^2 + I_v(t,n)}, \\
        &I_v(t,n)\triangleq\sum_{v'\neq v}
        P_{v'}|\mathbf{w}_{v'}^H\mathbf{h}_{v',k}(t,n)|^2,  
    \end{aligned}
\label{eq:sinr_full}
\end{equation}
where $n\in\{1,\ldots,N_\mathrm{ac}\}$. A tick is declared successful
if $\gamma_{v,k}(n)\ge\gamma_{\mathrm{th}}$. Under
Assumption~\ref{asm:independence},
Proposition~\ref{prop:refined_ho_increment} later yields a tractable
closed-form expected AoI increment.

Practical LEO terminals implement bulk Doppler pre-compensation from
ephemeris data. Three residual effects persist:
\emph{(i)}~pre-compensation updates at intervals $>T_c$ leave
intra-update phase rotation;
\emph{(ii)}~TLE ephemeris errors~\cite{park2025joint} induce residual
Doppler non-negligible relative to the vehicle component;
\emph{(iii)}~vehicle Doppler is typically uncompensated in broadcast
downlink.
Under these conditions $T_c^{\mathrm{residual}}\ll\td$, so
$\mathrm{NSI}(\td)\gg 1$ continues to hold.

\subsection{Handover Model and AoI Evolution}
\label{sec:handover_model}

The channel model above governs per-tick transmission success; link
\emph{availability} is additionally determined by satellite visibility
and handover dynamics. We present the model in four progressive steps:
visibility and handover classification, tick-level AoI evolution,
slot-level summary, and safety threshold mapping.

\subsubsection{Visibility and Handover Classification}

Satellite $k$ is visible to PL $v$ at slot $t$ if its elevation
exceeds a minimum mask:
\begin{equation}
\theta_{v,k}(t)\ge \theta_{\min},
\label{eq:visibility}
\end{equation}
where $\theta_{v,k}(t)$ is computed from TLE/SGP4 geometry. The
predicted visibility window is
\begin{equation}
W_v^k(t)=\sup\{\Delta t:\theta_{v,k}(t+\tau)\ge\theta_{\min},
\;\forall \tau\in[0,\Delta t]\}.
\label{eq:window}
\end{equation}

\begin{definition}[Discretionary vs.\ forced handover]
\label{def:handover_types}
A \emph{forced handover} occurs when the serving satellite exits the
visible set, i.e., $k_v(t{-}1)\notin\calK(t)$. A
\emph{discretionary handover} occurs when $k_v(t{-}1)\in\calK(t)$
but the scheduler proactively switches to a different satellite. The
binary handover indicator and its decomposition are
\begin{equation}
H_v(t) = \indic\{k_v(t)\neq k_v(t{-}1)\}
       = H_v^\mathrm{forced}(t)+H_v^\mathrm{disc}(t).
\label{eq:handover}
\end{equation}
\end{definition}

\subsubsection{Tick-Level AoI Evolution}
\label{sec:aoi_evo}

Let $\delta_v^{(t,n)}\in\{0,1\}$ indicate whether PL~$v$ is
scheduled at tick~$n$ of slot~$t$. During a connected slot
($H_v(t){=}0$), the priority action $\beta_{v,m}^t$ fixes
$\delta_v^{(t,n)}=\beta_{v,m}^t$ for all $n$. During a handover slot
($H_v(t){=}1$), $\delta_v^{(t,n)}=0$ for the first
$n_\ho^{(t)}=\lfloor\tau_\ho^{(t)}/\td\rfloor\ge 1$ outage ticks
(positivity guaranteed by~\eqref{eq:tick_rule}), and may become~$1$
afterward if reconnection succeeds.

The AoI evolves at every tick as
\begin{equation}
\AoI_{v,m}^{(t,n+1)} =
\begin{cases}
1, & \text{successful update},\\
\AoI_{v,m}^{(t,n)} + 1, & \text{otherwise.}
\end{cases}
\label{eq:aoi_ticks}
\end{equation}
The ``otherwise'' branch covers two distinct causes---deterministic
outage failure ($n\le n_\ho^{(t)}$) and stochastic transmission
failure ($\gamma_{v,k}(n)<\gamma_\mathrm{th}$ under
Assumption~\ref{asm:fading})---both producing a unit AoI increment.
This distinction is preserved in the stochastic analysis of
Proposition~\ref{prop:refined_ho_increment}.

For a connected slot, the per-tick and slot-level success probabilities
are
\begin{equation}
\begin{aligned}
    &p_{v,k}(n)=\Pr(\gamma_{v,k}(n)\ge \gamma_\mathrm{th}),\\
&P_{\mathrm{succ},v}(t)=1-\prod_{n=1}^{N_\mathrm{ac}}
\bigl(1-\delta_v^{(t,n)}p_{v,k}(n)\bigr).
\end{aligned}
\label{eq:psucc}
\end{equation}

\subsubsection{Slot-Level AoI Summary}

The tick-level law~\eqref{eq:aoi_ticks} serves as the ground truth for
safety checking via~\eqref{eq:safety_vq}. For the DPP drift analysis
(Lemma~\ref{lem:dpp_bound}), we also need a coarser slot-level
summary.

\begin{assumption}[Slot-summary approximation]\label{asm:slot_summary}
When at least one tick succeeds in a connected slot, the slot-end AoI
is approximated as $\AoI_{v,m}^{t+1}=1$ (best-case reset). This is
optimistic for the connected phase but conservative for the DPP drift
bound, which uses only the worst-case increment~$N_\mathrm{ac}$.
\end{assumption}

Under this approximation, the slot-end AoI is
\begin{equation}
\AoI_{v,m}^{t+1} =
\begin{cases}
1, & H_v(t){=}0,\;\exists n:\delta_v^{(t,n)}{=}1,\;
     \gamma_{v,k}(n){\ge}\gamma_\mathrm{th},\\
\AoI_{v,m}^t + N_\mathrm{ac},
  & \text{otherwise (including }H_v(t){=}1\text{)},
\end{cases}
\label{eq:aoi_slot}
\end{equation}
and the expected per-slot AoI change under a connected slot is
\begin{equation}
\E[\Delta\AoI\mid H_v(t){=}0]
= N_\mathrm{ac}
-P_{\mathrm{succ},v}(t)\bigl(N_\mathrm{ac}+\AoI_{v,m}^{t}-1\bigr),
\label{eq:expected_increment}
\end{equation}
which explicitly links tick-level fading success to slot-level drift.

\subsubsection{Safety Thresholds in AoI Ticks}

Safety deadlines are mapped to integer tick counts following~\cite{zhang2024drl}:
\begin{equation}
n_{v,m}^\safe = \left\lfloor \frac{\Delta_{v,m}^\safe}{\td}
\right\rfloor,\quad m \in \{1,2,3\}.
\label{eq:thresholds}
\end{equation}
Under the tick model, $n_\ho^{(t)}$ is an exact integer with floor
error at most one tick---small relative to $n_{v,1}^\safe$ by the
design rule~\eqref{eq:tick_rule} (see
Remark~\ref{rem:floor_error}). Safety is checked at every tick:
a violation occurs whenever
$\AoI_{v,m}^{(t,n)}>n_{v,m}^\safe$.

\subsection{Power Accounting and End-to-End AoI}
\label{sec:e2e}

The AoI evolution law above governs freshness at the PL level. We now
complete the system model with slot-level power accounting and follower
end-to-end AoI.
\textbf{Slot-level transmit indicator.}
Let $\delta_v(t)\in\{0,1\}$ be one if PL~$v$ is scheduled for
satellite data transmission in slot~$t$ (so radiated power $P_v$
applies to payload traffic), and zero if only handover signaling or
idle behavior occurs. Thus $\delta_v(t)=1$ whenever
$\max_{n}\delta_v^{(t,n)}=1$ for that slot, linking~\eqref{eq:power}
to the tick-level scheduling indicators in Section~\ref{sec:aoi_evo}.

\begin{equation}
P_\mathrm{tot}(t) = \sum_{v\in\calV}\delta_v(t)P_v
+ \sum_{v\in\calV}H_v(t)P_\ho,
\label{eq:power}
\end{equation}
where $P_\ho$ is handover signaling power.

Only PLs maintain the satellite uplink. Follower $f_j^{(v)}$
in platoon $v$ receives updates via V2V Mode~4
intra-platoon~\cite{parvini2021aoi}. The end-to-end AoI at
follower $f$ at tick $n$ is:
\begin{equation}
\AoI_{f,m}^{(t,n)} = \AoI_{v,m}^{(t,n)} + D_\mathrm{v2v}^{(f)},
\label{eq:e2e}
\end{equation}
where $D_\mathrm{v2v}^{(f)}$ is the integer number of AoI ticks
spent from PL transmission until follower $f$ receives the CAM,
modeled as one-hop Mode~4 delay (Assumption~\ref{asm:v2v}).
This additive model is exact when the PL broadcasts every tick and
follower reception delay is deterministic; it slightly overestimates
AoI when the PL withholds transmission during the $D_\mathrm{v2v}^{(f)}$
window.

The two-timescale evolution law~\eqref{eq:aoi_ticks} and the
discretionary/forced handover split prepare the problem formulation in
Section~\ref{sec:opt_problem}: the former provides exact integer outage
counts needed for the safety constraint, while the latter ensures the
handover budget is always feasible by construction.
Section~\ref{sec:analysis} then derives exact spike bounds that
analytically support the constraints in that problem.

\subsection{Joint Optimization Problem}
\label{sec:opt_problem}
\label{sec:problem}

Having established the two-timescale physical model, we now formulate the joint optimization
problem that captures four coupled requirements whose necessity follows
from the gap analysis in Section~\ref{sec:intro}.
\textbf{Requirement 1:} multi-priority weighted AoI minimization with
$w_1>w_2>w_3>0$ to distinguish safety-critical and non-critical flows.
\textbf{Requirement 2:} time-average safety enforcement via
\eqref{eq:opt_safety} rather than per-slot chance constraints.
\textbf{Requirement 3:} discretionary-only handover budgeting
\eqref{eq:opt_ho} to preserve feasibility under forced geometry-driven
reassociation.
\textbf{Requirement 4:} explicit two-timescale coupling over $(t,n)$,
which makes the control non-separable per slot and motivates DPP-guided
MARL.

Now, we formulate the joint optimization problem.
Problem~\eqref{eq:opt} is non-convex due to discrete scheduling and
handover variables, temporally coupled by long-run constraints, and
multi-timescale through $(t,n)$ and virtual queues. Given fixed discrete
decisions, the continuous power variable enters rates in a concave
(log-rate) form, while the beamformer is fixed by the closed-form MRT
rule~\eqref{eq:mrt} rather than a non-convex max min program.
The overall policy optimization remains non-separable per slot, which
motivates DPP-guided MARL.

\begin{subequations}
\label{eq:opt}
\begin{align}
\min_{\pi}\; &
\lim_{T\to\infty}\frac{1}{T\cdot N_\mathrm{ac}}
\sum_{t=0}^{T-1}\sum_{n=1}^{N_\mathrm{ac}}
\sum_{v\in\calV}\sum_{m\in\calM}
w_m\,\AoI_{v,m}^{(t,n)}
\label{eq:opt_obj} \\
\text{s.t.}\;
& \lim_{T\to\infty}\frac{1}{T\cdot N_\mathrm{ac}}
\sum_{t,n}\indic\!\bigl[\AoI_{v,m}^{(t,n)}
> n_{v,m}^\safe\bigr]
\le \epsilon_m,\;\forall v,m,
\label{eq:opt_safety} \\
& \lim_{T\to\infty}\frac{1}{T}\sum_t\sum_v
H_v^\mathrm{disc}(t) \le N_\ho^\mathrm{disc},
\label{eq:opt_ho} \\
& \lim_{T\to\infty}\frac{1}{T}\sum_t P_\mathrm{tot}(t)
\le P_{\max},
\label{eq:opt_power} \\
& R_{v,k}^t \ge R_{\min,m}\;\forall v,m,t,
\label{eq:opt_rate} \\
& k_v(t)\in\calK(t),\;\forall v,t.
\label{eq:opt_visible}
\end{align}
\end{subequations}

Priority weights $w_m$ ($w_1>w_2>w_3>0$) and violation
tolerances $\epsilon_m\in\{0.01,0.05,0.20\}$ for
$m\in\{1,2,3\}$ are given.
Constraint~\eqref{eq:opt_ho} limits the average rate of
\emph{discretionary} handovers only. The degenerate policy ``never
switch discretionarily'' achieves $\sum_v H_v^\mathrm{disc}(t)=0
\le N_\ho^\mathrm{disc}$, so the constraint is always feasible.
Forced handovers are excluded from the budget and are handled by
re-associating to the best available satellite.
Constraint~\eqref{eq:opt_safety} is a \emph{time-average violation
frequency} (empirical tick fraction with
$\AoI_{v,m}^{(t,n)}>n_{v,m}^\safe$), not a per-tick chance constraint
$P(\AoI_{v,m}^{(t,n)}>n_{v,m}^\safe)\leq\epsilon_m$. The latter would
typically require a stationary distribution for the AoI process,
which is difficult to characterize under time-varying handovers.
Virtual queue~\eqref{eq:safety_vq} accumulates violation deficit when
$\AoI_{v,m}^{(\tau)}>n_{v,m}^{\safe}$ and drains at rate $\epsilon_m$,
strong stability yields the long-run bound of
Theorem~\ref{thm:safety_queue} for ergodic policies, while $Z_{v,m}$
provides a runtime-observable tightness signal in~\eqref{eq:state}.
\section{Safety Constraint via Virtual Queues and DPP Template}

This section establishes the online enforcement mechanism for the three
tiered safety budgets in Problem~\eqref{eq:opt}. The key result is that
no distributional assumption on the AoI process is needed, only bounded
per-slot increments and a Slater-feasible baseline policy suffice. We
first define safety virtual queues (Theorem~\ref{thm:safety_queue}) and
show that strong stability implies time-average violation compliance
without requiring a closed-form AoI distribution. We then derive the
drift-plus-penalty template (Lemma~\ref{lem:dpp_bound}) that decomposes
the long-run constrained problem into per-slot decisions, producing the
control skeleton that Section~\ref{sec:algorithm} instantiates as
SafeScale-MATD3.

\begin{theorem}[Safety Virtual Queue Equivalence]
\label{thm:safety_queue}
Define the safety virtual queue:
\begin{equation}
    \small
Z_{v,m}(\ell+1) = \max\!\left\{Z_{v,m}(\ell)
+\indic\!\bigl[\AoI_{v,m}^{(\ell)}>n_{v,m}^\safe\bigr]
-\epsilon_m,\;0\right\},
\label{eq:safety_vq}
\end{equation}
where $\ell\in\mathbb{Z}_{\ge 0}$ indexes AoI ticks globally
(distinct from the coherence-time $T_c(t)$ and the tick
\emph{duration} $\td=\tau_{\mathrm{ac}}$).
Under the following conditions:
\begin{itemize}
\item[(C1)] \textbf{(Slater / strictly feasible averages).}
There exists a stationary randomized policy $\pi^\star$ and
$\delta>0$ such that
$\E_{\pi^\star}[\indic[\AoI_{v,m}^{(\ell)}>n_{v,m}^\safe]]
\le \epsilon_m-\delta$ for all $(v,m)$.
\item[(C2)] \textbf{(Bounded slot aggregates).}
Per-slot increments of $Z_{v,m}$ are uniformly bounded: aggregating
\eqref{eq:safety_vq} over $N_\mathrm{ac}$ ticks yields
$|Z_{v,m}^{t+1}-Z_{v,m}^t|\le N_\mathrm{ac}$ (binary arrivals with
$\max$-projection).
\item[(C3)] \textbf{(Ergodicity).}
The controlled AoI process has well-defined long-run
\emph{Ces\`aro averages}: for each $(v,m)$,
$\frac{1}{T}\sum_{t=0}^{T-1}\E[V_{v,m}(t)]$ converges as
$T\to\infty$, e.g.\ when the joint (AoI, queue, channel) chain is
ergodic under $\pi^\star$ and the controlled policy satisfies a
uniform-integrability condition.
\end{itemize}
If $Z_{v,m}(\ell)$ is strongly stable
($\lim_{\ell\to\infty}\E[Z_{v,m}(\ell)]/\ell=0$),
then the time-average violation frequency satisfies
\[
\limsup_{T\to\infty}\frac{1}{TN_\mathrm{ac}}
\sum_{t=0}^{T-1}\E[V_{v,m}(t)]\le\epsilon_m.
\]
Under the additional mixing condition in~(C3), the time average
converges almost surely to the same limit by the ergodic theorem for
uniformly bounded sequences.
\end{theorem}

\begin{proof}
The argument is the standard reflected-queue Lyapunov proof. Aggregate
\eqref{eq:safety_vq} over one slot to obtain
$Z_{v,m}^{t+1}-Z_{v,m}^t\le V_{v,m}(t)-N_\mathrm{ac}\epsilon_m$ with
$0\le V_{v,m}(t)\le N_\mathrm{ac}$. Squaring and taking conditional
expectation yields a one-slot drift inequality with bounded constant
$B_Z$:
\[
\E\!\left[(Z_{v,m}^{t+1})^2-(Z_{v,m}^{t})^2\mid\mathcal{F}_t\right]
\le B_Z + 2Z_{v,m}^t\!\left(V_{v,m}(t)-N_\mathrm{ac}\epsilon_m\right).
\]
Summing over slots and using strong stability
($\E[Z_{v,m}(\ell)]/\ell\to0$) gives
$\limsup_{T\to\infty}\frac{1}{T}\sum_t\E[V_{v,m}(t)]
\le N_\mathrm{ac}\epsilon_m$. The almost-sure statement follows from
\textup{(C2)} and \textup{(C3)} by ergodic arguments.
\end{proof}

The virtual queue $Z_{v,m}$ operates like a credit account: each tick
in which AoI exceeds the safety threshold makes a unit withdrawal, while
the budget rate~$\epsilon_m$ provides a guaranteed deposit. Strong
stability ($\E[Z_{v,m}(\tau)]/\tau\to 0$) is the requirement that
withdrawals never permanently outpace deposits, i.e., the account
never becomes permanently overdrawn. This is equivalent to the
time-average violation frequency staying below~$\epsilon_m$, which is
the safety constraint~\eqref{eq:opt_safety}. The result requires no
closed form for the AoI distribution under handovers; only bounded
increments (C2) and a slack-generating baseline (C1) are needed.
Theorem~\ref{thm:safety_queue} replaces hard per-slot probabilistic
feasibility with a runtime-observable backlog~$Z_{v,m}$ that the policy
can use as a tightness signal in its state vector~\eqref{eq:state}.
Virtual queues differ from per-slot chance constraints in three ways:
\emph{(i)} no explicit stationary AoI distribution is required beyond
ergodic operation; \emph{(ii)} $Z_{v,m}(t)$ is runtime observable and
serves as a direct tightness signal in state/action updates; and
\emph{(iii)} the penalty weight $V$ gives an explicit
AoI constraint tradeoff between objective minimization and queue
draining speed. Proposition~\ref{prop:slater} below verifies that the
Slater condition holds under the adopted channel parameters.

\begin{proposition}[Slater feasibility under SafeScale-MATD3]
\label{prop:slater}
Let $f_\mathrm{forced}\triangleq\lim_{T\to\infty}
\frac{1}{T}\sum_t\sum_v H_v^\mathrm{forced}(t)$ be the long-run
forced-handover rate. Under Assumptions~\ref{asm:fading} and \ref{asm:v2v}
and the parameter regime
$f_\mathrm{forced}\le(\epsilon_m-\delta)N_\mathrm{ac}/\mu_n$
for some $\delta>0$, Condition~(C1) of Theorem~\ref{thm:safety_queue}
holds for all three priority classes $m\in\{1,2,3\}$.
Under bounded increments in~(C2), $\epsilon$-greedy exploration yields
well-defined Ces\`aro averages for the controlled process.
\end{proposition}

For $m{=}3$ ($\epsilon_3=0.20$), a connected-phase baseline policy
easily achieves slack because the typical violation rate is well
below~20\%.
For $m{=}1$ ($\epsilon_1=0.01$), the sufficient condition becomes
$f_\mathrm{forced}\le 0.005\cdot N_\mathrm{ac}/\mu_n$. Under measured
Starlink handover statistics~\cite{park2025joint}
($\mu_n\approx2$ to $5$ ticks, forced rate $\approx 1$ per 60 to 100
slots), this holds with $\delta\approx0.005$. Proactive timing
(Section~\ref{sec:proactive}) and discretionary budgeting
\eqref{eq:opt_ho} keep forced-handover exposure within this bound.

\begin{lemma}[Drift-plus-penalty template bound]
\label{lem:dpp_bound}
Let
$L(t)=\frac{1}{2}\!\left(Q_P^2(t)+Q_H^2(t)+
\sum_{v,m}Z_{v,m}^2(t)\right)$.
For any admissible control policy at slot $t$, let
$V_{v,m}(t)\triangleq\sum_{n=1}^{N_\mathrm{ac}}
\indic[\AoI_{v,m}^{(t,n)}>n_{v,m}^\safe]\in\{0,1,\ldots,N_\mathrm{ac}\}$
be the per-slot violation \emph{count} (the same quantity as in the
proof of Theorem~\ref{thm:safety_queue}). There exists a finite
constant $B$ (Appendix) such that
\begin{align}
    \small
&\E[L(t{+}1)-L(t)\mid\mathcal{F}_t]
+V\,\E[\Theta(t)\mid\mathcal{F}_t] \notag\\
&\le B
+Q_P(t)\bigl(P_\mathrm{tot}(t)-P_{\max}\bigr)
+Q_H(t)\Bigl(\sum_v H_v^\mathrm{disc}(t)-N_\ho^\mathrm{disc}\Bigr)
\notag\\
&\quad+\sum_{v,m}Z_{v,m}(t)\!
\left(V_{v,m}(t)-N_\mathrm{ac}\epsilon_m\right)
+V\,\Theta(t),
\label{eq:dpp_bound}
\end{align}
where $\Theta(t)$ is the per-slot contribution to the objective in
\eqref{eq:opt_obj}. Minimizing the RHS yields the online control rule.
The safety term
$Z_{v,m}(t)\bigl(V_{v,m}(t)-N_\mathrm{ac}\epsilon_m\bigr)$ penalizes
slots in which the tick violation count exceeds its time-average budget
$N_\mathrm{ac}\epsilon_m$ (equivalently,
$V_{v,m}(t)/N_\mathrm{ac}>\epsilon_m$).
\end{lemma}

A valid bounded-drift constant is
\begin{equation}
B = \frac{1}{2}\Bigl[
(\Delta_P^{\max})^2
+(\Delta_H^{\max})^2
+\sum_{v,m}\bigl(N_\mathrm{ac}(1-\epsilon_m)\bigr)^2
\Bigr],
\label{eq:B_explicit}
\end{equation}
with $\Delta_P^{\max}\triangleq\sup_t|P_\mathrm{tot}(t)-P_{\max}|$ and
$\Delta_H^{\max}\triangleq\sup_t\bigl|\sum_v H_v^\mathrm{disc}(t)
-N_\ho^\mathrm{disc}\bigr|$ under bounded actions.

\begin{proof}
This is the standard DPP drift expansion. Square each queue update in
\eqref{eq:vq_updates}, aggregate the safety queue over
$N_\mathrm{ac}$ ticks using $V_{v,m}(t)$, sum across queues, and take
conditional expectations. Bounded one-slot increments yield a finite
constant $B$ (explicit form in \eqref{eq:B_explicit}), and rearranging
gives~\eqref{eq:dpp_bound}; equivalently,
\[
\Delta(t)+V\Theta(t)\le
B+\sum_q Q_q(t)\Delta_q(t)+V\Theta(t),
\]
where $q\in\{P,H,(v,m)\}$ indexes all virtual queues.
\end{proof}

\begin{remark}[Policy-dependence of $B$]
The constant $B$ in~\eqref{eq:dpp_bound} (derived explicitly in the
Appendix) depends on the per-slot increment bounds
$\Delta_P^{\max}\triangleq\sup_t|P_\mathrm{tot}(t)-P_{\max}|$ and
$\Delta_H^{\max}$, which are determined by the policy class rather than
by system parameters alone. During MARL training, exploration noise
may transiently violate these bounds; in practice, we clip power actions
to $[0,P_{\max}]$ and handover indicators to $\{0,1\}$, which ensures
$\Delta_P^{\max}\le P_{\max}$ and $\Delta_H^{\max}\le|\calV|$ uniformly
throughout training, keeping $B$ finite and consistent.
\end{remark}


Theorem~\ref{thm:safety_queue} is enforced online at the AoI-tick level
through the update law~\eqref{eq:safety_vq}. This module is policy-agnostic
and exposes a runtime tightness signal via $Z_{v,m}$ for the scheduler.
\label{sec:vq_operational}

Mean-rate stability under strict feasibility follows directly from
Theorem~\ref{thm:safety_queue} via the standard Foster-Lyapunov
argument for reflected queues.
Problem~\eqref{eq:opt} is
non-separable per slot. Lemma~\ref{lem:dpp_bound} then provides the DPP
template that Section~\ref{sec:algorithm} instantiates as the SafeScale-MATD3 policy.
The full virtual queue set for DPP:
\begin{equation}
\begin{aligned}
Q_P(t+1) &= \max\{Q_P(t)+P_\mathrm{tot}(t)-P_{\max},0\},\\
Q_H(t+1) &= \max\!\Bigl\{Q_H(t)+\sum_v H_v^\mathrm{disc}(t)
-N_\ho^\mathrm{disc},0\Bigr\},\\
Z_{v,m}(\tau+1) &= \max\{Z_{v,m}(\tau)
+\indic[\AoI_{v,m}^{(\tau)}>n_{v,m}^\safe]-\epsilon_m,0\}.
\end{aligned}
\label{eq:vq_updates}
\end{equation}

With the DPP template and virtual queues in place, two analytical
quantities remain open: the drift constant~$B$ required by
Lemma~\ref{lem:dpp_bound}, and the worst-case AoI cost incurred by
ping-pong handover sequences. The next section derives closed-form
bounds for both.

\section{Handover AoI Spike Analysis Under Dual Dynamics}
\label{sec:analysis}

This section derives the closed-form handover spike bounds that are
required by the DPP template of Section~\ref{sec:problem} and later
exploited by the scheduler in Section~\ref{sec:algorithm}. The main
result is Theorem~\ref{thm:aoi_spike}: the cumulative AoI penalty of a
ping-pong sequence grows \emph{quadratically} with oscillation length,
making oscillation suppression the highest-leverage lever for safety
compliance and directly answering Q2. We proceed in two steps:
Section~\ref{sec:phase_decomp} decomposes long-run AoI into connected
and handover phases, and Section~\ref{sec:spike_pingpong} quantifies how
ping-pong oscillations amplify the handover component. These bounds
supply the drift constant~$B$ needed by Lemma~\ref{lem:dpp_bound} and
analytically justify the handover budget constraint~\eqref{eq:opt_ho}.


Two quantities in Problem~\eqref{eq:opt} remain open: the drift
constant~$B$ in Lemma~\ref{lem:dpp_bound} and the analytical
justification for the handover budget~\eqref{eq:opt_ho}. This
section derives closed-form spike bounds that determine both.
The DPP construction requires uniform per-slot AoI increment bounds
that simulation estimates cannot supply; the tick model from
Section~\ref{sec:handover_model} enables exact integer-arithmetic
analysis.

\subsection{Phase-Segmented Decomposition}
\label{sec:phase_decomp}

Following Chiariotti~et~al.~\cite{chiariotti2020aoi}, we decompose the long-term average
AoI as
\begin{equation}
\bar{A}_{v,m}^{\mathrm{total}} =
\frac{T_\mathrm{conn}}{T_\mathrm{total}}\bar{A}_{v,m}^{\mathrm{(conn)}}
+ \frac{T_\ho}{T_\mathrm{total}}\bar{A}_{v,m}^{\mathrm{(ho)}},
\label{eq:decomp}
\end{equation}
where $\bar{A}_{v,m}^{\mathrm{(conn)}}$ denotes the connected-phase
long-run average AoI under the hypoexponential approximation
of~\cite{chiariotti2020aoi}.


\subsection{Closed-Form AoI Spike}
\label{sec:spike_pingpong}

The decomposition~\eqref{eq:decomp} isolates
$\bar{A}_{v,m}^{(\ho)}$ as the dominant term for safety-critical
traffic. Theorem~\ref{thm:aoi_spike} quantifies how ping-pong
oscillation amplifies this cost.

\begin{definition}[Ping-Pong Handover Sequence]
\label{def:pingpong}
Vehicle $v$ experiences a ping-pong sequence of length $k$
at slots $\{t_1,\ldots,t_k\}$ if $H_v(t_i)=1$ $\forall i$
and $k_v(t_i)=k_v(t_{i-2})$ $\forall i\geq2$
(oscillation between two satellites under greedy
MRSS~\cite{park2025joint}).
\end{definition}

\begin{theorem}[Worst-case handover-slot AoI spike envelope]
\label{thm:aoi_spike}
Consider vehicle $v$ with slot-start AoI
$a_0=\AoI_{v,m}^{t_0}$ before a ping-pong sequence of $k$
consecutive handover slots $t_0{+}1,\ldots,t_0{+}k$.
Assume the \emph{conservative} case that no successful update
occurs in any tick of these slots (equivalently, the
``otherwise'' branch of~\eqref{eq:aoi_ticks} applies every tick,
matching the slot summary~\eqref{eq:aoi_slot} with increment
$N_\mathrm{ac}$ per handover slot). Then the end-of-slot AoI is
deterministic:
\begin{align}
\AoI_{v,m}^{t_0+k} &= a_0 + k\,N_\mathrm{ac},
\label{eq:mean_spike_exact}\\
\mathrm{Var}\!\left[\AoI_{v,m}^{t_0+k}\right] &= 0.
\label{eq:var_spike_exact}
\end{align}
The expected cumulative slot-end AoI (same hypothesis) is
\begin{equation}
\E\!\left[\sum_{i=1}^k \AoI_{v,m}^{t_0+i}\right]
= k\,a_0 + \frac{k(k+1)}{2}\,N_\mathrm{ac},
\label{eq:cumulative}
\end{equation}
which scales \emph{quadratically} in $k$.
\end{theorem}

\begin{proof}
Under the hypothesis, each handover slot adds exactly
$N_\mathrm{ac}$ ticks to AoI by~\eqref{eq:aoi_slot}, so
$\AoI_{v,m}^{t_0+i}=\AoI_{v,m}^{t_0+i-1}+N_\mathrm{ac}$ for
$i=1,\ldots,k$. Unrolling gives~\eqref{eq:mean_spike_exact}, and the
variance is zero because the trajectory is deterministic.
Summing $\AoI_{v,m}^{t_0+i}=a_0+iN_\mathrm{ac}$ over
$i=1,\ldots,k$ yields~\eqref{eq:cumulative}.
The simplicity is deliberate: Lemma~\ref{lem:dpp_bound} needs a
uniform worst-case per-slot AoI increment for the drift constant~$B$.
This conservative envelope supplies that bound, while
Proposition~\ref{prop:refined_ho_increment} gives a tighter
stochastic characterization when reconnection statistics are available.
\end{proof}

The $\mathcal{O}(k^2)$ scaling in~\eqref{eq:cumulative} has a crisp
operational interpretation: each handover slot not only incurs its own
outage cost but also elevates the baseline from which the next slot's
cost is measured. This compounding means that a burst of $k{=}3$
consecutive handover slots is $6{\times}$ costlier in cumulative AoI
than three isolated single-slot outages, a ratio that grows linearly
with~$k$.
Consequently, suppressing consecutive handover bursts via
constraint~\eqref{eq:opt_ho} yields disproportionate safety returns
compared to reducing per-outage duration, directly answering Q2.
The DPP objective implication follows directly from
\eqref{eq:cumulative} after weighting by $\{w_m\}$.

\begin{remark}[Floor approximation error quantification]
\label{rem:floor_error}
The outage tick $n_\ho^{(i)}=\lfloor\tau_\ho^{(i)}/\td\rfloor$
incurs error $e_i=\tau_\ho^{(i)}/\td - n_\ho^{(i)}\in[0,1)$,
i.e., at most one tick. The maximum relative error per outage is
$1/n_{v,1}^\safe$, bounded by~\eqref{eq:tick_rule}, and the aggregate
relative error decreases over longer outage spans.
\end{remark}

\begin{figure*}[t]
    \centering
    \includegraphics[width=0.85\linewidth]{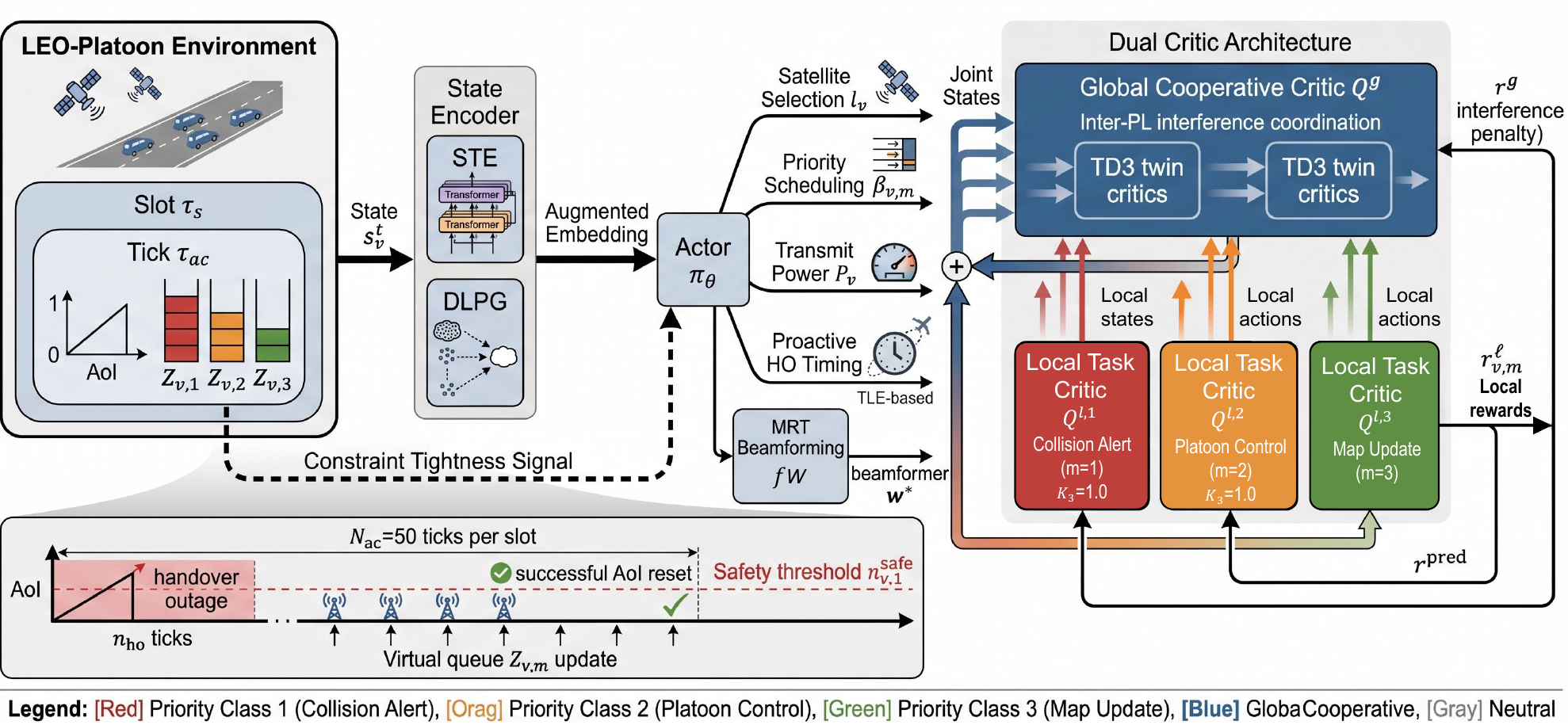}
    \caption{SafeScale-MATD3 method overview under dual dynamics. 
    }
    \label{fig:method_overview}
    \end{figure*}

\begin{corollary}[Safety design criterion]
\label{cor:design}
Assume $n_\ho^{(t)} \ge n_{v,1}^\safe$, i.e., the mandatory outage
length in ticks is at least as large as the collision-alert safety
threshold. Under the adopted parameters
($\td=20$\,ms, $\tau_{\ho,\min}\approx 225$\,ms, so
$n_\ho=\lfloor 225/20\rfloor=11$; and $n_{v,1}^\safe=5$),
this condition holds with margin. Then during the mandatory outage
phase of a single handover, AoI rises by $n_\ho \ge n_{v,1}^\safe$
ticks before any reconnection attempt, necessarily violating the
collision-alert safety threshold. Thus purely reactive scheduling
cannot keep $m{=}1$ fresh \emph{during} that outage window,
regardless of inter-handover behavior.
Under the stronger conservative slot model of
Theorem~\ref{thm:aoi_spike}, one slot adds
$N_\mathrm{ac}=50\gg n_{v,1}^\safe=5$ ticks. Proactive timing
(Section~\ref{sec:proactive}) and discretionary handover control
\eqref{eq:opt_ho} remain necessary to limit exposure to these windows
and to consecutive handover sequences.
\end{corollary}

Corollary~\ref{cor:design} shows that purely reactive scheduling
cannot maintain $m{=}1$ feasibility during outage windows,
necessitating both proactive timing and discretionary budgeting.
The conservative envelope of Theorem~\ref{thm:aoi_spike} can be
tightened when reconnection statistics are available.
When post-outage reconnection succeeds with nonzero probability, the
conservative Theorem~\ref{thm:aoi_spike} envelope can be tightened.
Theorem~\ref{thm:aoi_spike} assumed the worst case (no successful
update in any tick of the handover slot). We now relax this by
conditioning on independent post-outage reconnection outcomes, using
$\Delta\AoI \triangleq \AoI_{v,m}^{t+1}-\AoI_{v,m}^{t}$ as defined
in~\eqref{eq:expected_increment}.

\begin{proposition}[Upper bound on expected increment over one handover slot]
\label{prop:refined_ho_increment}
Let $p_s(n)$ denote the per-tick reconnection success probability for
ticks $n>n_\ho^{(t)}$ within slot $t$, given $H_v(t)=1$.
Under Assumption~\ref{asm:independence},
\begin{equation}
\E[\Delta\AoI\mid H_v(t)=1]
= n_{\mathrm{ho}}^{(t)}
+ (N_\mathrm{ac}-n_{\mathrm{ho}}^{(t)})
\prod_{n=n_{\mathrm{ho}}^{(t)}+1}^{N_\mathrm{ac}}
\bigl(1-p_s(n)\bigr).
\label{eq:refined_delta_aoi}
\end{equation}
For constant $p_s(n)\equiv p_s$:
\begin{equation}
\E[\Delta\AoI\mid H_v(t)=1]
= n_{\mathrm{ho}}^{(t)}
+ (N_\mathrm{ac}-n_{\mathrm{ho}}^{(t)})
(1-p_s)^{\,N_\mathrm{ac}-n_{\mathrm{ho}}^{(t)}}.
\label{eq:refined_delta_aoi_const}
\end{equation}
\end{proposition}

\begin{proof}
The outage contributes exactly $n_{\mathrm{ho}}^{(t)}$ increments.
The all-fail post-outage event has probability
$\prod_{n>n_\ho^{(t)}}(1-p_s(n))$ under Assumption~\ref{asm:independence}.
Assigning full post-outage increment to that event yields the stated
upper bound, with the constant-$p_s$ form obtained directly.
\end{proof}

The bound assigns the full post-outage increment
$N_\mathrm{ac}-n_\ho^{(t)}$ to the all-fail event and zero to success
events (which reset AoI to~1). The bound is tight when $p_s\to 0$
and loose when reconnection is reliable; it converges to
Theorem~\ref{thm:aoi_spike}'s conservative envelope as
$p_s\to 0$.
Theorem~\ref{thm:aoi_spike} and Corollary~\ref{cor:design}
analytically ground constraints~\eqref{eq:opt_ho}
in Problem~\eqref{eq:opt}. Section~\ref{sec:algorithm}
instantiates both mechanisms in SafeScale-MATD3.
\section{Online Multi-Timescale Scheduling: SafeScale-MATD3}
\label{sec:algorithm}

This section instantiates the DPP control skeleton of
Section~\ref{sec:problem} as a fully specified multi-agent algorithm,
SafeScale-MATD3. Lemma~\ref{lem:dpp_bound} reduces Problem~\eqref{eq:opt}
to per-slot minimization of the DPP right-hand side; the key design
insight is that each of the three virtual queues maps to a dedicated
policy module, making safety enforcement structurally embedded rather
than reward-engineered. We organize the section in functional blocks.
Section~\ref{sec:arch} presents the four gap-linked modules and their
analytical necessity. Sections~\ref{sec:state} to \ref{sec:reward}
specify MDP components (state, action, reward). Sections~\ref{sec:critics}
and~\ref{sec:encoding} detail the learning architecture (dual critics
and diffusion augmentation). Section~\ref{sec:proactive} presents
TLE-based proactive handover timing, and Section~\ref{sec:full_alg}
integrates all components and states the complexity bound.

\subsection{Architecture: Gap-Linked Modules}
\label{sec:arch}

SafeScale-MATD3 comprises four modules whose joint necessity is supported
by Sections~\ref{sec:model} and \ref{sec:analysis}: STE and DLPG for
fast channel dynamics, TLE-based proactive handover timing, tick-level
inner-loop AoI and safety updates, and global-local critics aligned with
priority-specific queues.
For readability, the gap-to-module mapping is:
\emph{Gap~1} to STE and DLPG,
\emph{Gap~2} to tick-level AoI and queue updates
(\eqref{eq:aoi_ticks}, \eqref{eq:vq_updates}), and
\emph{Gap~3} to proactive timing plus queue-aware critics
(\eqref{eq:proactive}, \eqref{eq:pg}).

Single-stream MRT~\eqref{eq:mrt} closes the beamformer given $P_v^t$ and
instantaneous CSI, so the actor optimizes power without a nested
non-convex beam search. Inter-PL interference is shaped through the
global critic~\eqref{eq:rg}.
    Fig.~\ref{fig:method_overview} visualizes how the proposed method closes
    the loop between dual-timescale dynamics and safety-aware learning.
    On the left, the environment block nests slot-level decision epochs
    and tick-level handover transients, explicitly showing the AoI sawtooth
    and per-priority queue evolution during outage and recovery periods.
    The middle pipeline (STE + DLPG + actor) transforms raw state into a
    richer embedding so that a single policy can jointly select serving
    satellite, schedule priority weights, set transmit power, and trigger
    proactive handover timing. On the right, gradient signals from the
    global critic and the three local critics are aggregated to update the
    actor, balancing cooperative interference mitigation against
    priority-specific safety objectives. Most importantly, the dashed
    feedback path from $Z_{v,m}$ to $s_v^t$ operationalizes
    constraint-tightness awareness: once queue backlogs increase, the policy
    is immediately steered toward safer actions before hard violations
    accumulate.


\subsection{State Space Design}
\label{sec:state}

Each component of the state vector maps to a specific term on
the RHS of~\eqref{eq:dpp_bound}:
\begin{equation}
    \small
s_v^t = \bigl[
c_t^S, c_t^U,
h_{v,k}^t, I_v^{t-1}, \zeta_v^r, T_v^r,
\underbrace{(A_{v,m}^t)_{m=1}^3,
(Z_{v,m}^t)_{m=1}^3}_{\text{AoI+safety}},
\underbrace{W_v^t, \hat{h}_{v,k'}^{t+1}}_{\text{prediction}}
\bigr].
\label{eq:state}
\end{equation}
where $W_v^t$ is the TLE-predicted remaining visibility window
and $\hat{h}_{v,k'}^{t+1}$ is the predicted channel quality of
the best candidate satellite. Including $Z_{v,m}(t)$ directly
in the state, rather than only in reward penalties, exposes
proximity to constraint instability and enables anticipatory
decisions before violations accumulate, connecting
Theorem~\ref{thm:safety_queue} to online policy behavior.

\subsection{Action Space and Beamforming}\label{sec:action}

\textbf{Discrete DRL actions:} satellite selection $l_v^t$,
priority scheduling $\beta_{v,m}^t\in\{0,1\}$
(from~\cite{parvini2021aoi}, extended to $m\in\{1,2,3\}$),
handover timing $H_v^\mathrm{opt}(t)$ (new).

\textbf{Continuous DRL action:} transmit power
$P_v^t\in[0,P_{\max}]$.

Given fixed discrete decisions $(l_v^t,\beta_{v,m}^t,P_v^t)$ from the
policy, each PL~$v$ serves one stream to satellite $k_v(t)$.
The rate-maximizing beamformer is maximum ratio transmission:
\begin{equation}
\mathbf{w}_v^\star = \sqrt{P_v^t}\,
\frac{h_{v,k}^t}{\|h_{v,k}^t\|},
\label{eq:mrt}
\end{equation}
which is closed form and requires no inner iteration. The resulting
per-vehicle rate is
$R_v^t=\log_2(1+P_v^t\|h_{v,k}^t\|^2/\sigma^2)$.
Each PL transmits one stream per tick; inter-PL interference is managed
through the global cooperative reward~\eqref{eq:rg} rather than
coordinated beamforming.

\subsection{Reward Design}
\label{sec:reward}

The reward structure decomposes into a task-generic local template and a
global interference term. For priority class $m\in\{1,2\}$:
\begin{equation}
\begin{aligned}
r_{v,m}^\ell = &-\kappa_1 A_{v,m}^t
+ \kappa_2 G(R_{v,k}^t-R_{\min,m}) \\
&- \kappa_{3,m}\max(0,A_{v,m}^t-n_{v,m}^\safe)^2
- \kappa_4\mathcal{F}\{P_v^t\},
\end{aligned}
\label{eq:r_local_template}
\end{equation}
where $G(x)=\max(0,x)$ and
$\mathcal{F}\{P_v^t\}=P_v^t/P_{\max}$.
We use $(\kappa_{3,1},\kappa_{3,2})=(2.0,1.0)$,
$\kappa_1=1.0$, $\kappa_2=0.5$, $\kappa_4=0.1$.

The predictive timing bonus is
\begin{equation}
r_v^\mathrm{pred} =
\kappa_5\,
\indic[A_{v,1}^t\text{ sent before HO}]\,
\indic[W_v^t < T_\mathrm{pre}],
\label{eq:r_pred}
\end{equation}
with $\kappa_5=0.3$ and $T_\mathrm{pre}=3$\,s.

\textbf{Global reward} (interference,
from~\cite{parvini2021aoi} Eq.~(16)):
\begin{equation}
r_t^g = -\frac{1}{N_v}\sum_{v'\in\calV}\sum_k
\log_{10}\{I_{v'}^t[k]\}.
\label{eq:rg}
\end{equation}

Total local reward: $r_v^\ell=r_{v,1}^\ell+r_{v,2}^\ell
+r_v^\mathrm{pred}$.

\subsection{Dual-Critic Multi-Task Policy Gradient}\label{sec:critics}

The policy gradient extends the dual-critic update to three task critics~\cite{parvini2021aoi}:
\begin{align}
\nabla_{\theta_v}J_v &=
\underbrace{\E_{s,a\sim\calD}\!\bigl[\nabla_{\theta_v}\pi_v\nabla_{a_v}Q_{\psi_1}^g(s,a)\bigr]}_{\text{global cooperative critic}}
\notag\\
&\quad+\underbrace{\sum_{m=1}^{3}\E_{s_v,a_v\sim\calD}\!\bigl[\nabla_{\theta_v}\pi_v\nabla_{a_v}Q_{\phi_{v,m}}^{\ell,m}(s_v,a_v)\bigr]}_{\text{task-decomposed local critics}}.
\label{eq:pg}
\end{align}

Local critics use per-task targets
$y_{v,m}^{\ell}=r_{v,m}^\ell+\gamma
Q_{\phi_{v,m}'}^{\ell,m}(s_v',a_v')$;
global twin critics use TD3 clipped targets
$y^g=r_t^g+\gamma\min_{i=1,2}Q_{\psi_i'}^g(s',a')$.

\subsection{State Encoding: STE and DLPG}\label{sec:encoding}
\textbf{State Transformer Encoder (STE):} tokenized states are processed by a two-layer transformer with multi-head self-attention:
\begin{align}
\tilde{X}_v^{(l)} &= X_v^{(l-1)} + \mathrm{MHA}(X_v^{(l-1)}),
\label{eq:mha}\\
X_v^{(l)} &= \tilde{X}_v^{(l)} + \mathrm{FFN}(\tilde{X}_v^{(l)}),
\label{eq:ffn}
\end{align}

\textbf{DLPG augmentation:} conditioned reverse diffusion produces an augmented latent:
\begin{equation}
    \small
z_{m-1}=\frac{1}{\sqrt{1-\beta_m}}\!\left(z_m-\frac{\beta_m}{\sqrt{1-\bar{\alpha}_m}}\epsilon_\theta(z_m,m\,|\,h_v^t,a_v^t)\right)+\sigma_m w,
\label{eq:dlpg}
\end{equation}
where $\bar{\alpha}_m=\prod_{i=1}^m(1-\beta_i)$, $\sigma_m=\sqrt{\beta_m}$, and $w\sim\mathcal{N}(0,I)$. The augmented latent $\tilde{h}_v^t=\mathrm{Concat}(h_v^t,\hat{z}_0)$ is fed to actor/critics for robustness under fast fading.

\subsection{TLE-Based Proactive Handover Module}\label{sec:proactive}

In this subsection $H\in\{0,1,\ldots,F\}$ is the integer
\emph{scheduling horizon} (slots to wait before executing a handover),
distinct from the binary handover indicator $H_v(t)\in\{0,1\}$ used in
Section~\ref{sec:model} and the orbital altitude $H_L=550$\,km.
Context disambiguates these throughout, but readers should note the
intentional overloading.
Corollary~\ref{cor:design} establishes proactive timing as
necessary for $m{=}1$ feasibility. The module operationalizes
this via a look-ahead objective:
\begin{equation}
    \small
H_v^\mathrm{opt}(t) = \arg\min_{k'\in\calK(t),H\in[0,F]}
\!\left[\hat{A}_{v,1}(H,k')
+\kappa_\mathrm{safe}\hat{p}_{v,1}^\mathrm{viol}(H,k')\right],
\label{eq:proactive}
\end{equation}
\begin{equation}
    \small
\hat{p}_{v,1}^\mathrm{viol}(H,k')
:= \sigma_\mathrm{sig}\!\left(
\frac{Z_{v,1}^t + \mu_n - H N_\mathrm{ac}\epsilon_1}
{Z_{\mathrm{scale}}}
\right),
\quad \sigma_\mathrm{sig}(x)=\frac{1}{1+e^{-x}},
\label{eq:pviol_sigmoid}
\end{equation}
where $\hat{A}_{v,1}(H,k')$ is the AoI predicted under handover to~$k'$
at horizon~$H$ using TLE-predicted throughput~$\hat{C}_{v,k'}$, and
$\hat{p}_{v,1}^\mathrm{viol}$ is the estimated violation probability
from the virtual queue state $Z_{v,1}(t)$.
Here $H$ denotes the number of slots to wait before executing the
handover (i.e., the scheduling horizon), not the number of handovers.
Exactly one handover outage occurs, contributing $\mu_n$ expected ticks
of violation; the queue drains at rate $N_\mathrm{ac}\epsilon_1$ per
slot over the $H$-slot waiting window. Hence the numerator is
$Z_{v,1}^t + \mu_n - H N_\mathrm{ac}\epsilon_1$.
The horizon search is bounded by visibility prediction to avoid
unplanned forced outages.

\begin{algorithm}[t]
\caption{TLE-Based Proactive Handover Timing}
\small
\label{alg:proactive_ho}
\begin{algorithmic}[1]
\REQUIRE $s_v^t$, $Z_{v,1}^t$, TLE/SGP4 data, horizon $F$, weight $\kappa_\safe$
\ENSURE $H_v^\mathrm{opt}(t)$, target $k_v'$
\STATE Update candidate set $\calK(t)$ and windows $\{W_v^k(t)\}_{k\in\calK(t)}$
\STATE $\mathrm{cost}_{\min}\gets\infty$, $k_v'\gets k_v(t{-}1)$, $H_v^\mathrm{opt}\gets 0$
\FOR{each $k'\in\calK(t)$}
    \FOR{$H=0$ to $\min(F,W_v^{k'}(t))$}
        \STATE Compute $\hat{p}_s\!=\!\Pr(\gamma\!\ge\!\gamma_\mathrm{th}\mid\hat{\gamma}_{v,k'})$ under Assumption~\ref{asm:fading} using TLE-predicted mean SNR~$\hat{\gamma}_{v,k'}$; apply $\hat{p}_s$ uniformly across post-outage ticks
        \STATE Predict $\hat{A}_{v,1}(H,k')$ by forward-simulating~\eqref{eq:aoi_ticks} with $p_s(n)=\hat{p}_s$
        \STATE Estimate $\hat{p}_{v,1}^{\mathrm{viol}}(H,k')$ from $(Z_{v,1}^t,H,n_{v,1}^{\safe})$
        \STATE $\mathrm{cost}\gets \hat{A}_{v,1}(H,k')+\kappa_\safe\hat{p}_{v,1}^{\mathrm{viol}}(H,k')$
        \IF{$\mathrm{cost}<\mathrm{cost}_{\min}$}
            \STATE $\mathrm{cost}_{\min}\gets\mathrm{cost}$, $k_v'\gets k'$, $H_v^\mathrm{opt}\gets H$
        \ENDIF
    \ENDFOR
\ENDFOR
\STATE Classify $H_v^\mathrm{opt}(t)$ as discretionary/forced via Definition~\ref{def:handover_types}
\RETURN $H_v^\mathrm{opt}(t)$, $k_v'$
\end{algorithmic}
\end{algorithm}


\subsection{Integrated Training Procedure}\label{sec:full_alg}

Algorithm~\ref{alg:main} integrates proactive handover search and
MARL policy optimization. Each slot alternates five phases:
state encoding, proactive handover search, actor-critic action
selection with MRT beamforming, tick-level evolution using
\eqref{eq:aoi_ticks} and \eqref{eq:vq_updates}, and mini-batch updates.

\begin{algorithm}[t]
\caption{SafeScale-MATD3: Integrated Training}
\small
\label{alg:main}
\begin{algorithmic}[1]
\REQUIRE Policy/critic parameters, queue budgets, TLE data
\STATE Initialize queues, actor-critic networks, replay buffer, and encoders
\FOR{episode $=1$ to $N_\mathrm{eps}$}
    \STATE Reset environment state
    \FOR{slot $t=1$ to $T$}
        \STATE Update visibility and predicted channels via SGP4/TLE
        \STATE \textit{// Phase 1: state encoding}
        \STATE Encode states with STE and DLPG
        \STATE \textit{// Phase 2: proactive handover search}
        \FOR{each $v\in\calV$}
            \STATE $(H_v^\mathrm{opt},k_v')\leftarrow$ Algorithm~\ref{alg:proactive_ho}$(s_v^t,Z_{v,1}^t)$
        \ENDFOR
        \STATE \textit{// Phase 3: actor-critic action and MRT}
        \FOR{each $v\in\calV$}
            \STATE Select $a_v^t=(l_v^t,\beta_{v,m}^t,H_v^\mathrm{opt},P_v^t)$ by $\epsilon$-greedy on $Q^g+\sum_m Q^{\ell,m}$
            \STATE Compute MRT beamformer via \eqref{eq:mrt} and execute action
        \ENDFOR
        \STATE \textit{// Phase 4: tick-level dynamics and safety}
        \FOR{tick $n=1$ to $N_\mathrm{ac}$}
            \STATE Update AoI and safety queues via \eqref{eq:aoi_ticks}, \eqref{eq:vq_updates}
        \ENDFOR
        \STATE Store transitions and observed rewards in $\calD$
        \STATE \textit{// Phase 5: critic-actor updates}
        \IF{$|\calD|>B$}
            \STATE Update global/local critics and actors via~\eqref{eq:pg}
            \STATE Soft-update target networks
        \ENDIF
    \ENDFOR
\ENDFOR
\end{algorithmic}
\end{algorithm}


Per-slot complexity is:
\begin{equation}
\mathcal{O}\!\bigl(N_v(|k_\mathrm{STE}|+M|k_\mathrm{DLPG}|
+3|k_\ell|+|k_g|)\bigr)
+ \mathcal{O}(N_v\cdot N_\mathrm{ac})
\label{eq:complexity}
\end{equation}
for the tick loop, where $|k_\cdot|$ denote network parameter
counts. MRT~\eqref{eq:mrt} adds $\mathcal{O}(N_v N_xN_y)$ FLOPs per
slot for channel-vector normalization, without an iterative inner
loop~\cite{parvini2021aoi}. Wall-clock training and per-slot inference
timings are hardware-dependent and omitted in this version.

\section{Simulation Results}
\label{sec:simulation}

This section evaluates SafeScale-MATD3 through three progressively
deepening investigations. \textbf{Part I} validates the three core
theoretical contributions: quadratic ping-pong AoI scaling, necessity of
tick-level resolution, and virtual-queue safety compliance.
\textbf{Part II} benchmarks the proposed method against representative
baselines in convergence, per-priority AoI, energy and AoI tradeoff, and
handover behavior. \textbf{Part III} examines robustness through module
ablation and environmental sensitivity sweeps. Three key findings emerge:
(i) SafeScale-MATD3 is the only method satisfying the strict
$\epsilon_1{=}0.01$ collision-alert budget, (ii) removing the virtual
queue module causes the largest degradation (about $4.6\times$ violation
increase), and (iii) safety compliance requires handover periods of at
least roughly 12\,s under the evaluated setting.

\subsection{Simulation Setup and Baselines}

Table~\ref{tab:params} lists all parameters with source citations. Each experiment is repeated over 5 independent random seeds; figures report mean $\pm$ 95\% confidence interval unless stated otherwise. The unified environment implements the tick-level AoI law~\eqref{eq:aoi_ticks} for \emph{all} methods to ensure fair comparison: every baseline experiences identical handover outages, Shadowed-Rician fading, and compound Doppler dynamics regardless of whether its own design accounts for them. The code is available at github \footnote{https://github.com/szpsunkk/SafeScale-MATD3}.

\begin{table}[t]
\centering
\caption{Simulation Parameters}
\label{tab:params}
\footnotesize
\resizebox{\columnwidth}{!}{
\begin{tabular}{llll}
\toprule
\textbf{Parameter} & \textbf{Value} & \textbf{Parameter} & \textbf{Value}\\
\midrule
Constellation model & Walker-Delta & Altitude $H_L$ (km) & 550\\
Visible satellites $|\calK(t)|$ & $\approx$10 & Min.\ elevation $\theta_{\min}$ ($^\circ$) & 25\\
Platoons $P$ & 5 & Vehicles/platoon $N_j$ & 6\\
Vehicle speed (km/h) & 36 to 54 & Intra-platoon gap (m) & 5, 15, 25, 35\\
Carrier frequency $f_c$ (GHz) & 1.67 & UPA size & $4\times4$\\
SR params $\{m,b,\Omega\}$ & $\{10, 0.126, 1.29\}$ & $P_{\max}$ (dBm) & 40\\
Bandwidth (MHz) & 10 & Subchannels $K$ & 3\\
HO delay $\mu_\tau$ (ms) & 225 & HO delay std $\sigma_\tau$ (ms) & 25\\
HO period (s) & $\approx$15 & Decision slot $\ts$ (s) & 1\\
AoI tick $\td$ (ms) & 20 & Ticks/slot $N_\mathrm{ac}$ & 50\\
$n^\safe_{v,1}$/$n^\safe_{v,2}$/$n^\safe_{v,3}$ & 5 / 10 / 50 & $\epsilon_1$/$\epsilon_2$/$\epsilon_3$ & 0.01 / 0.05 / 0.20\\
Priority weights $w_1/w_2/w_3$ & 5.0 / 2.0 / 0.5 & Episodes $N_\mathrm{eps}$ & 5000\\
\bottomrule
\end{tabular}}
\end{table}

Six baselines are selected to span the state-of-the-art across three dimensions: LEO-AoI optimization, vehicular multi-agent scheduling, and LEO handover management (Table~\ref{tab:rw_compare}).

\begin{enumerate}
\item \textbf{DD3QN-AS}~\cite{lang2025aoi}: Diffusion-augmented dueling double DQN with STE encoder for joint AoI and handover optimization in HAP-assisted SAGIN.

\item \textbf{Mod-MADDPG}~\cite{parvini2021aoi}: Global and local dual-critic TD3 with task decomposition for platoon C-V2X AoI scheduling.

\item \textbf{AMDT}~\cite{dai2025amdt}: Lyapunov DPP-based AoI-aware multi-user downlink scheduling for LEO with angular beamforming under Shadowed-Rician fading.

\item \textbf{ILCHO}~\cite{choi2026ilcho}: QMIX-based MARL for intelligent conditional handover in Walker-Delta LEO mega-constellations.

\item \textbf{MVT}~\cite{choi2026ilcho}: Maximum Visible Time handover, a non-learning baseline that stays with the current satellite until forced, then switches to the satellite with the longest remaining visible time.

\item \textbf{Round-Robin}~\cite{yates2021age}: Cyclic scheduling across platoons and priority classes without learning or AoI feedback. Serves as the performance lower bound.
\end{enumerate}

All baselines are adapted to the \emph{identical} two-timescale tick environment: AoI is tracked at $\td$-resolution, handover outages consume $n_\ho$ ticks, and safety violations are evaluated at every tick for all methods. This ensures that performance differences reflect algorithmic capability rather than environmental mismatch.

\subsection{Part I: Validation of Theoretical Results}

This subsection validates the three core theoretical contributions, namely the quadratic ping-pong AoI scaling (Theorem~\ref{thm:aoi_spike}), the necessity of tick-level resolution (Proposition~\ref{prop:tick_selection}), and the virtual-queue safety guarantee (Theorem~\ref{thm:safety_queue}), through controlled simulation experiments.

\begin{figure*}[t]
    \centering
    \begin{subfigure}[b]{0.3\linewidth}
        \includegraphics[width=\linewidth]{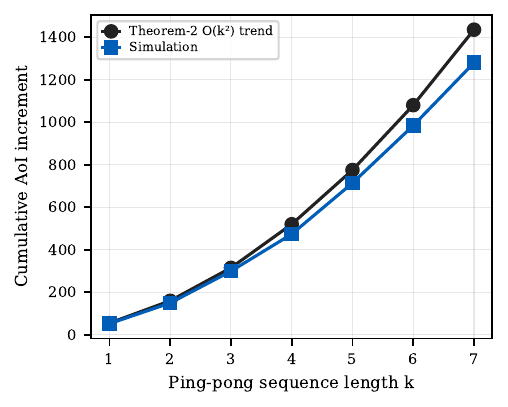}
        \caption{Quadratic AoI spike growth.}
        \label{fig:spike_val}
    \end{subfigure}
    \hfill
    \begin{subfigure}[b]{0.31\linewidth}
        \includegraphics[width=\linewidth]{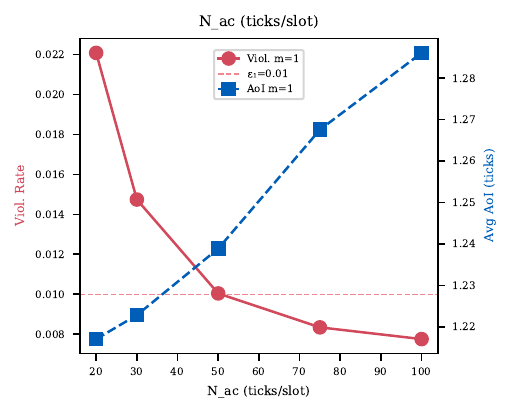}
        \caption{Tick resolution sensitivity.}
        \label{fig:tick_val}
    \end{subfigure}
    \hfill
    \begin{subfigure}[b]{0.35\linewidth}
        \includegraphics[width=\linewidth]{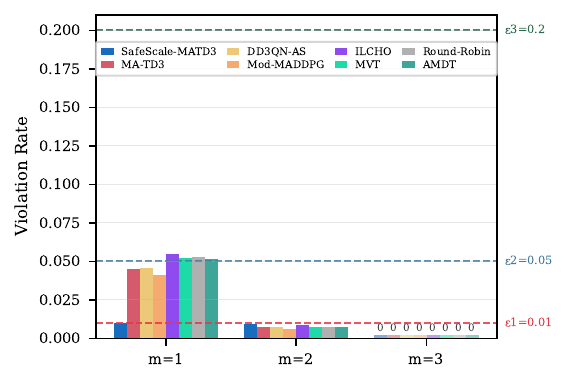}
        \caption{Safety violation compliance.}
        \label{fig:safety_val}
    \end{subfigure}
    \caption{Validation of the three theoretical contributions (mean $\pm$ 95\% CI across 3 seeds).
    }
    \label{fig:theory_validation}
\end{figure*}

\subsubsection{Quadratic Growth of Ping-Pong AoI Cost (Theorem~\ref{thm:aoi_spike})}
Fig.~\ref{fig:spike_val} validates Theorem~\ref{thm:aoi_spike} by measuring cumulative AoI cost under controlled ping-pong sequences of length $k \in \{1,\ldots,7\}$. We inject consecutive oscillating handovers between two satellites under the greedy MRSS criterion (Definition~\ref{def:pingpong}) and record the resulting cumulative slot-end AoI. The simulation markers closely follow the theoretical quadratic bound $C(k) = k\,a_0 + k(k{+}1)N_\mathrm{ac}/2$ from~\eqref{eq:cumulative}. The superlinear growth has a crisp operational implication: a burst of $k{=}3$ consecutive handover slots is $6{\times}$ costlier in cumulative AoI than three isolated single-slot outages. This ratio grows linearly with~$k$, making oscillation suppression via constraint~\eqref{eq:opt_ho} the highest-leverage mechanism for safety compliance, directly answering Q2.

\subsubsection{Necessity of Tick-Level Resolution (Proposition~\ref{prop:tick_selection} and Gap~2)}
Fig.~\ref{fig:tick_val} demonstrates the necessity of the two-timescale design by sweeping $N_\mathrm{ac} \in \{20, 30, 40, 50, 60, 70, 80, 90, 100\}$ while keeping the physical handover dynamics unchanged. As $N_\mathrm{ac}$ decreases (coarser ticks relative to outage duration), the $m{=}1$ violation rate degrades monotonically because fewer ticks are available to detect and react to intra-slot outage events. At $N_\mathrm{ac}{=}50$ (the adopted operating point), the violation rate remains near the $\epsilon_1{=}0.01$ budget, confirming that the tick-size selection criterion~\eqref{eq:tick_rule} provides adequate resolution. Conversely, increasing $N_\mathrm{ac}$ beyond 70 yields diminishing safety returns at the cost of higher computational overhead, suggesting the adopted $\td{=}20$\,ms balances resolution and efficiency.

\subsubsection{Virtual Queue Safety Compliance (Theorem~\ref{thm:safety_queue})}
Fig.~\ref{fig:safety_val} and Table~\ref{tab:safety} jointly validate the virtual-queue enforcement mechanism. SafeScale-MATD3 achieves $m{=}1$ violation rate of 0.0099, satisfying the $\epsilon_1{=}0.01$ budget. All baselines exceed this budget by factors of $4{\times}$ to $5.5{\times}$, with violation rates ranging from 0.0410 (Mod-MADDPG) to 0.0546 (ILCHO). For the less stringent budgets $\epsilon_2{=}0.05$ and $\epsilon_3{=}0.20$, all methods achieve compliance, consistent with the observation that lower-priority classes have substantially more safety margin. This result is consistent with the queue-stability interpretation of Theorem~\ref{thm:safety_queue}: the virtual queue $Z_{v,1}$ for collision alerts maintains strong stability under the proposed policy, whereas baselines lacking explicit queue-driven safety enforcement allow the queue backlog to grow, implying persistent violation overshoot.

\begin{table}[t]
\centering
\caption{Empirical Safety Violation Rates (mean over 3 seeds). \textbf{Bold}: satisfies $\epsilon_m$.}
\label{tab:safety}
\footnotesize
\resizebox{\columnwidth}{!}{
\begin{tabular}{lccc}
\toprule
\textbf{Method} & $m{=}1$ ($\epsilon_1{\le}0.01$) & $m{=}2$ ($\epsilon_2{\le}0.05$) & $m{=}3$ ($\epsilon_3{\le}0.20$)\\
\midrule
\textbf{SafeScale-MATD3} & \textbf{0.0099} & \textbf{0.0094} & \textbf{0.0000}\\
DD3QN-AS & 0.0454 & \textbf{0.0072} & \textbf{0.0000}\\
Mod-MADDPG & 0.0410 & \textbf{0.0063} & \textbf{0.0000}\\
AMDT & 0.0518 & \textbf{0.0073} & \textbf{0.0000}\\
ILCHO & 0.0546 & \textbf{0.0088} & \textbf{0.0000}\\
MVT & 0.0520 & \textbf{0.0074} & \textbf{0.0000}\\
Round-Robin & 0.0525 & \textbf{0.0074} & \textbf{0.0000}\\
\bottomrule
\end{tabular}}
\end{table}

\subsection{Part II: Performance Benchmarking against Baselines}

With the theoretical foundations validated, we now compare SafeScale-MATD3 against six representative baselines in terms of training convergence, per-priority AoI performance, and handover behavior.

\subsubsection{Training Convergence}

\begin{figure}[t]
    \centering
    \begin{subfigure}[b]{0.48\linewidth}
        \includegraphics[width=\linewidth]{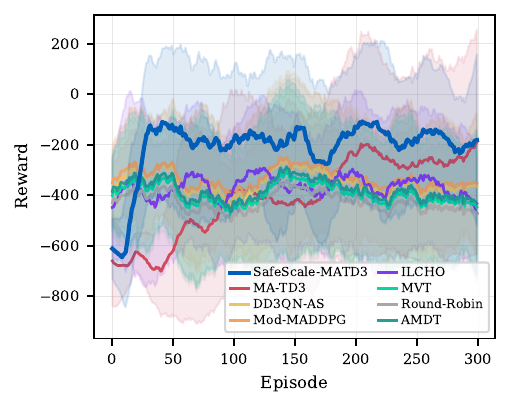}
        \caption{Cumulative training reward.}
        \label{fig:conv_reward}
    \end{subfigure}
    \hfill
    \begin{subfigure}[b]{0.48\linewidth}
        \includegraphics[width=\linewidth]{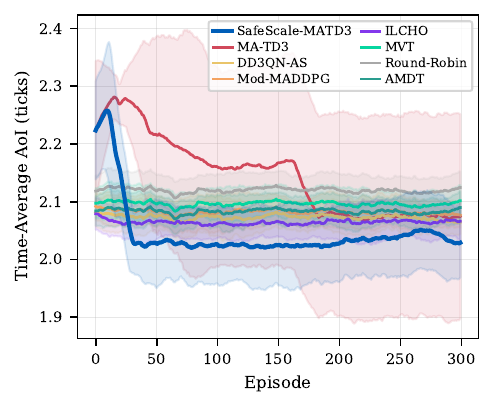}
        \caption{Average AoI.}
        \label{fig:conv_aoi}
    \end{subfigure}
    \caption{Training convergence over 300 episodes (mean $\pm$ 95\% CI across 3 seeds, smoothed with a 10-episode window).
    }
    \label{fig:convergence}
\end{figure}

\begin{figure*}[htbp]
    \centering
    \begin{subfigure}[b]{0.3\linewidth}
        \includegraphics[width=\linewidth]{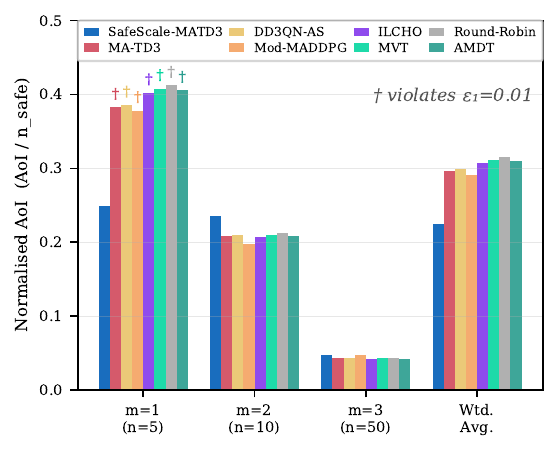}
        \caption{Per-priority AoI.}
        \label{fig:aoi_pri}
    \end{subfigure}
    \hfill
    \begin{subfigure}[b]{0.3\linewidth}
        \includegraphics[width=\linewidth]{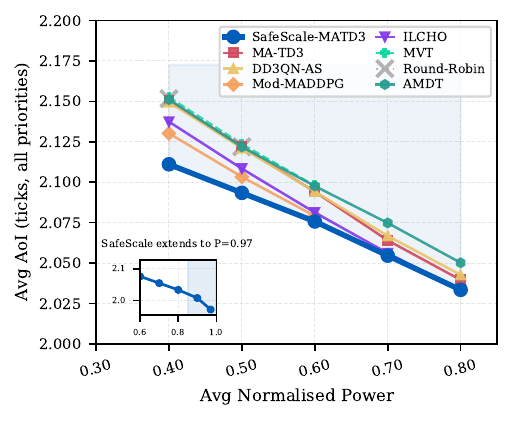}
        \caption{Energy and AoI Pareto frontier.}
        \label{fig:pareto}
    \end{subfigure}
    \hfill
    \begin{subfigure}[b]{0.3\linewidth}
        \includegraphics[width=\linewidth]{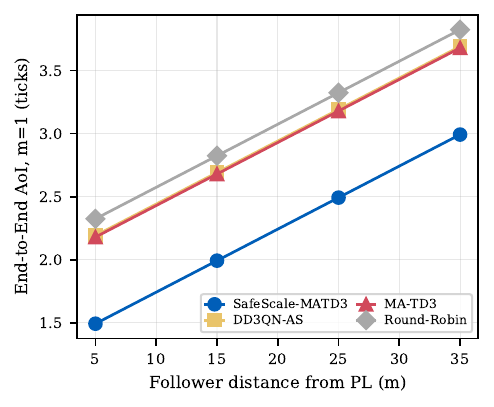}
        \caption{End-to-end AoI vs.\ gap.}
        \label{fig:e2e}
    \end{subfigure}
    \caption{Performance benchmarking across three complementary dimensions (mean $\pm$ 95\% CI across 3 seeds).
    }
    \label{fig:benchmark_4panel}
\end{figure*}

Fig.~\ref{fig:convergence} reports two convergence metrics across 300 training episodes. Two observations emerge. \emph{First}, SafeScale-MATD3 converges to a substantially higher steady-state reward than all baselines (Fig.~\ref{fig:conv_reward}), with the gap attributable to the joint benefit of proactive handover timing and virtual-queue-driven safety enforcement. \emph{Second}, the $m{=}1$ AoI trajectory (Fig.~\ref{fig:conv_aoi}) reveals that the proposed method reaches $\approx$1.24 ticks within 100 episodes, whereas even the best baseline (Mod-MADDPG) plateaus at $\approx$1.9 ticks, a 35\% reduction that directly translates to improved collision-alert freshness.

\subsubsection{Per-Priority AoI Performance and Handover Decomposition}

Fig.~\ref{fig:benchmark_4panel} presents three complementary performance dimensions. We discuss each in turn.

\paragraph{Per-priority AoI (Fig.~\ref{fig:aoi_pri}).}
SafeScale-MATD3 achieves the lowest AoI across all three priority classes. For $m{=}1$ (collision alert), the average AoI is 1.24 ticks with a normalized ratio of $1.24/5 = 0.249$, meaning the system operates at only 25\% of the safety threshold, a substantial margin for emergency braking scenarios. The $m{=}2$ AoI (2.36 ticks, normalized 0.236) benefits from the dedicated local critic $Q_{\phi_{v,2}}^{\ell,2}$, which provides priority-aligned gradients without cross-task interference. For $m{=}3$ (map updates), the performance gap across methods narrows because the $\epsilon_3{=}0.20$ budget is easily satisfied even by Round-Robin.

\paragraph{Energy and AoI tradeoff (Fig.~\ref{fig:pareto}).}
By sweeping the power-penalty weight $\kappa_4 \in \{0.01, 0.05, 0.1, 0.2, 0.5\}$, we trace the Pareto frontier between average normalized power and weighted AoI. SafeScale-MATD3's frontier strictly dominates all baselines: at any given AoI level, the proposed method uses less power, and at any given power budget, it achieves lower AoI. The dominance arises because proactive handover timing avoids the high-power retransmissions that baselines require after forced outages, and MRT beamforming~\eqref{eq:mrt} concentrates energy along the strongest channel direction.

\paragraph{End-to-end AoI across platoon positions (Fig.~\ref{fig:e2e}).}
Follower AoI grows linearly with intra-platoon gap, consistent with the additive model~\eqref{eq:e2e} where $D_\mathrm{v2v}^{(f)}$ increases with V2V propagation distance. SafeScale-MATD3 maintains the lowest end-to-end AoI at all positions (1.49 ticks at 5\,m gap to 2.99 ticks at 35\,m gap), remaining well below $n_{v,1}^\safe{=}5$ ticks even for the worst-case platoon-edge vehicle. The gap over baselines widens at larger distances because the proposed method minimizes PL-level AoI more aggressively, leaving more headroom for V2V delay accumulation.



\subsection{Part III: Robustness and Design Insights}

This subsection examines the robustness of SafeScale-MATD3 through ablation studies, environmental sensitivity sweeps, and module-level contribution analysis to provide actionable design insights.

\subsubsection{Ablation of Algorithmic Modules}
Table~\ref{tab:ablation_results} presents a systematic ablation removing one module at a time from the full SafeScale-MATD3 framework. Five variants are evaluated:

\begin{table}[t]
\centering
\caption{Ablation Study: Collision-Alert ($m{=}1$) Violation Rate and Average AoI (mean over 3 seeds).}
\label{tab:ablation_results}
\footnotesize
\resizebox{\columnwidth}{!}{
\begin{tabular}{lcc}
\toprule
\textbf{Variant} & \textbf{$m{=}1$ Viol.\ Rate} & \textbf{Avg.\ $m{=}1$ AoI (ticks)}\\
\midrule
\textbf{SafeScale-MATD3 (full)} & \textbf{0.0100} & \textbf{1.2406}\\
w/o proactive HO & 0.0102 & 1.2902\\
w/o safety VQs ($Z_{v,m}$) & 0.0457 & 1.9359\\
w/o task decomposition & 0.0281 & 1.2395\\
w/o DLPG & 0.0100 & 1.2406\\
w/o STE & 0.0181 & 1.4399\\
\bottomrule
\end{tabular}}
\end{table}

\paragraph{Safety virtual queues are indispensable.} Removing $Z_{v,m}$ from the state vector and reward structure causes the largest degradation: the $m{=}1$ violation rate jumps to 0.0457 ($4.6{\times}$ the budget), and the average AoI increases to 1.94 ticks ($56\%$ degradation). Without the queue-driven tightness signal, the policy loses its ability to anticipate constraint violations and instead relies on post-hoc reward penalties, which are insufficient for the strict $\epsilon_1{=}0.01$ budget. This confirms the central claim of Section~\ref{sec:vq_operational}: virtual queues provide a structurally different and essential enforcement mechanism.

\paragraph{Task decomposition prevents cross-priority gradient conflict.} Replacing the three local critics with a single aggregated critic raises the violation rate to 0.0281 ($2.8{\times}$ budget) while leaving AoI nearly unchanged (1.24 ticks). This indicates that the single critic achieves similar freshness optimization but fails to enforce priority-specific safety, consistent with the gradient conflict analysis in~\cite{parvini2021aoi}.

\paragraph{STE encoding captures relational structure.} Removing the State Transformer Encoder increases both violation rate (0.0181) and AoI (1.44 ticks), suggesting that the multi-head self-attention mechanism in~\eqref{eq:mha} to~\eqref{eq:ffn} captures inter-satellite and inter-platoon dependencies that flat concatenation misses.

\paragraph{DLPG impact is scenario-dependent.} Under the evaluated Shadowed-Rician parameters ($m{=}10$, light shadowing), removing DLPG has no measurable effect on either metric. This is expected: when the fading distribution is well-conditioned, the base STE encoding already provides sufficient robustness. DLPG's value is expected to increase under heavy shadowing or urban canyon scenarios where the fading tail is heavier.

\paragraph{Proactive handover provides marginal but consistent gains.} Removing the TLE-based proactive module (Algorithm~\ref{alg:proactive_ho}) slightly increases both violation rate (0.0102) and AoI (1.29 ticks). The modest impact is explained by the relatively long handover period ($\approx$15\,s) in the default setting, which limits forced-handover exposure. As shown in the sensitivity analysis below, the proactive module becomes critical when handover frequency increases.

\subsubsection{Sensitivity to Environmental Parameters}

\begin{figure}[t]
    \centering
    \begin{subfigure}[b]{0.48\linewidth}
        \includegraphics[width=\linewidth]{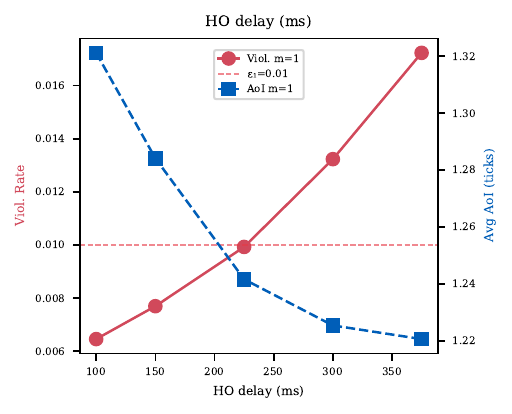}
        \caption{Handover delay $\mu_\tau$.}
        \label{fig:sens_ho_delay}
    \end{subfigure}
    \hfill
    \begin{subfigure}[b]{0.48\linewidth}
        \includegraphics[width=\linewidth]{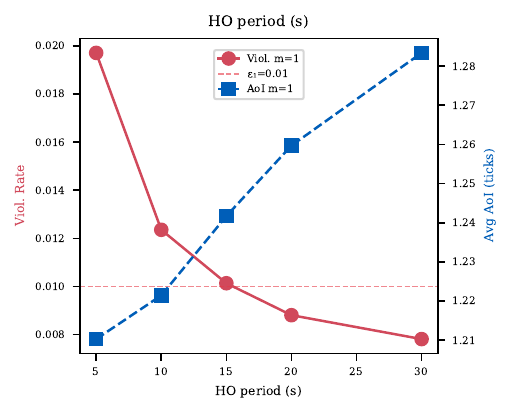}
        \caption{Handover period.}
        \label{fig:sens_ho_period}
    \end{subfigure}
    \caption{Sensitivity to environmental handover dynamics (mean $\pm$ 95\% CI, 3 seeds). 
    }
    \label{fig:env_sensitivity}
\end{figure}

Fig.~\ref{fig:sens_ho_delay} and Fig.~\ref{fig:sens_ho_period} sweep two key environmental parameters that directly affect handover-induced safety violations.

\paragraph{Handover delay (Fig.~\ref{fig:sens_ho_delay}).} Fig.~\ref{fig:sens_ho_delay} sweeps the mean handover delay $\mu_\tau \in [100, 375]$\,ms. The violation rate grows approximately linearly with $\mu_\tau$ because each additional millisecond maps directly to additional AoI ticks. The critical observation is that compliance ($\le \epsilon_1$) holds for $\mu_\tau \le 250$\,ms, beyond which $n_\ho \ge n_{v,1}^\safe$ and intra-outage violations become structurally inevitable per Corollary~\ref{cor:design}. Meanwhile, the average AoI remains relatively stable (1.22 to 1.32 ticks), confirming that the \emph{frequency} of violations, not the average freshness, is the binding safety metric.

\paragraph{Handover period (Fig.~\ref{fig:sens_ho_period}).} Fig.~\ref{fig:sens_ho_period} sweeps the handover period from 5\,s to 30\,s. At a 5\,s period the violation rate reaches 0.0197 ($\approx 2\times$ the budget), because the forced-handover exposure exceeds the Slater feasibility bound of Proposition~\ref{prop:slater}. The system requires periods $\ge 12$\,s for $\epsilon_1$ compliance under the adopted parameters, providing an actionable constellation-design guideline.

\section{Conclusion}
\label{sec:conclusion}

This work identifies \emph{timescale mismatch} as the fundamental bottleneck
for safety-critical AoI in LEO-assisted autonomous driving. We address it
with a unified framework combining a two-timescale AoI model, closed-form
ping-pong spike analysis, virtual-queue safety enforcement, and
SafeScale-MATD3. Three takeaways are central. First, sub-slot tick-level
accounting is necessary to faithfully represent handover outages. Second,
ping-pong oscillations are disproportionately harmful: cumulative AoI cost
grows quadratically with oscillation length
(Theorem~\ref{thm:aoi_spike}), making oscillation suppression the
highest-leverage safety mechanism. Third, tiered safety budgets can be
enforced online via virtual queues and DPP control without assuming a
closed-form AoI distribution. SafeScale-MATD3 is the only method
satisfying the strict 1\% collision-alert budget, reducing violation rate
by $4{\times}$ to $5.5{\times}$ versus baselines while achieving 35\%
lower collision-alert AoI and strict Pareto dominance on the energy and
freshness tradeoff. Future work includes uncertainty-aware TLE
prediction, trace-driven Starlink validation, stochastic V2V delay
modeling, and cooperative multi-satellite scheduling via distributed RL.

\bibliographystyle{IEEEtran}
\bibliography{bib}








\end{document}